\documentclass{article}

\usepackage{amsmath, amssymb}
\usepackage{fullpage}
\usepackage{amsthm, mathrsfs}
\usepackage{color}
\usepackage{hyperref}
\usepackage{slashed}
\usepackage{graphicx}
\usepackage{bbm}
\usepackage{accents}
\usepackage{cite}
\usepackage[affil-it]{authblk}

\newtheorem{theorem}{Theorem}[section]
\newtheorem{proposition}[theorem]{Proposition}

\newtheorem{remark}[theorem]{Remark}

\numberwithin{equation}{section}

\newcommand{\tr}{\mathrm{tr}}
\newcommand{\supp}{\mathrm{supp}}
\newcommand{\gslash}{\slashed{g}}
\newcommand{\Gammaslash}{\slashed{\Gamma}}

\begin{document}
\title{Future stability of expanding spatially homogeneous FLRW solutions of the spherically symmetric Einstein--massless Vlasov system with spatial topology $\mathbb{R}^3$}
\author{Martin Taylor}
\affil{\small Imperial College London,
Department of Mathematics,
South~Kensington~Campus,~London~SW7~2AZ,~United~Kingdom\vskip.2pc \ martin.taylor@imperial.ac.uk}
\date{June 30, 2023}

\maketitle

\begin{abstract}
	Spatially homogeneous FLRW solutions constitute an infinite dimensional family of explicit solutions of the Einstein--massless Vlasov system with vanishing cosmological constant.  Each member expands towards the future at a decelerated rate.  These solutions are shown to be nonlinearly future stable to compactly supported spherically symmetric perturbations, in the case that the spatial topology is that of $\mathbb{R}^3$.  The decay rates of the energy momentum tensor components, with respect to an appropriately normalised double null frame, are compared to those around Minkowski space.  When measured with respect to their respective $t$ coordinates, certain components decay faster around Minkowski space, while others decay faster around FLRW.
\end{abstract}

\tableofcontents

\section{Introduction}

Standard homogeneous isotropic cosmological models in general relativity are described by the Friedmann--Lema\^itre--Robertson--Walker (FLRW) spacetimes
\begin{equation} \label{eq:FLRWgeneral}
	\mathcal{M} = I \times \Sigma,
	\qquad
	g = -dt^2 + a(t)^2 g_{\Sigma},
\end{equation}
where $I\subset \mathbb{R}$ is an open interval, $(\Sigma,g_{\Sigma})$ is a constant curvature manifold, and $a\colon I \to (0,\infty)$ is an appropriate \emph{scale factor}.  See Section 5.\@3 of \cite{HaEl} for more on FLRW spacetimes.  This article concerns radiation filled FLRW cosmologies in which the constant curvature manifold is Euclidean space, $(\Sigma,g_{\Sigma}) = (\mathbb{R}^3,g_{\mathrm{Eucl}})$, with radiation described by spatially homogenous solutions of the massless Vlasov equation, and their stability properties.

\subsection{The Einstein--massless Vlasov system}

Consider a $3+1$ dimensional Lorentzian manifold $(\mathcal{M},g)$ and let
\[
	P = \{ (t,x,p) \in T\mathcal{M} \mid g_{(t,x)} (p, p) = 0 \} \subset T\mathcal{M},
\]
denote the \emph{mass shell} of $(\mathcal{M},g)$.
Consider some local coordinates $\{t,x^1,x^2,x^3\}$  on $\mathcal{M}$, and let $\{t,x^i, p^{\mu}\}$ denote the corresponding conjugate coordinate system for $T\mathcal{M}$, i.\@e.\@ $(t,x^i, p^{\mu})$ describes the point $p^{\mu} \partial_{x^{\mu}}\vert_{(t,x)} \in T \mathcal{M}$.  The massless Vlasov equation on $(\mathcal{M},g)$ takes the form
\begin{equation} \label{eq:EV1}
	p^0 \partial_t f + p^i \partial_{x^i} f - p^{\mu} p^{\nu} \Gamma_{\mu \nu}^i \partial_{p^i} f = 0,
\end{equation}
where $p^0$ is defined by the mass shell relation
\begin{equation} \label{eq:EV2}
	g_{\mu \nu} p^{\mu} p^{\nu} = 0.
\end{equation}
The Einstein--massless Vlasov system consists of equation \eqref{eq:EV1} coupled to the Einstein equations
\begin{equation} \label{eq:EV3}
	Ric(g)_{\mu\nu} - \frac{1}{2} R(g) g_{\mu\nu} = T_{\mu\nu},
\end{equation}
where the energy momentum tensor takes the form
\begin{equation}  \label{eq:EV4}
	T^{\mu \nu}(t,x)
	=
	\int_{P_{(t,x)}}
	f(t,x,p) p^{\mu} p^{\nu}
	\frac{\sqrt{- \det g }}{- p_0}
	dp^1 dp^2 dp^3,
\end{equation}
where indices are raised and lowered with respect to the metric $g$ (so that, for example, $p_0 = g_{0\mu} p^{\mu}$).  Here $t$ is also denoted $x^0$, Greek indices, such as $\mu, \nu$, range over $0,1,2,4$, and lower case Latin indices, such as $i,j,k$ range over $1,2,3$.

\subsection{The spatially homogeneous FLRW family of solutions}
\label{subsec:FLRWintro}

The \emph{spatially homogeneous FLRW} family constitutes an infinite dimensional family of explicit solutions of the Einstein--massless Vlasov system \eqref{eq:EV1}--\eqref{eq:EV4}.

Define the manifold
\begin{equation} \label{eq:manifold}
	\mathcal{M}_{\circ} = (0,\infty) \times \mathbb{R}^3,
\end{equation}
and consider the Euclidean metric
\[
	g_{\mathrm{Eucl}} = (dx^1)^2 + (dx^2)^2 + (dx^3)^2,
\]
on $\mathbb{R}^3$.  Given $x\in \mathbb{R}^3$ and $p = (p^1,p^2,p^3) \in T_x \mathbb{R}^3$, define
\[
	\vert p \vert^2
	=
	\vert p\vert_{g_{\mathrm{Eucl}}}^2
	=
	(g_{\mathrm{Eucl}})_{ij} p^i p^j
	=
	(p^1)^2 + (p^2)^2 + (p^3)^2.
\]
For any smooth, sufficiently decaying function $\mu : [0,\infty) \to [0,\infty)$ such that $\mu \not\equiv 0$, the metric $g_{\circ}$ and particle density $f_{\circ}$ defined by
\begin{equation} \label{eq:FLRW1}
	g_{\circ} = -dt^2 + a(t)^2 \big( (dx^1)^2 + (dx^2)^2 + (dx^3)^2 \big),
	\qquad
	f_{\circ}(t,x,p)
	=
	\mu (a(t)^4 \vert p \vert^2),
\end{equation}
where
\begin{equation} \label{eq:FLRW2}
	a(t) = t^{\frac{1}{2}} \left( \frac{4 \varrho}{3} \right)^{\frac{1}{4}},
	\qquad
	\varrho = \int \vert p \vert \mu(\vert p \vert^2) dp,
\end{equation}
define a solution of \eqref{eq:EV1}--\eqref{eq:EV4} on \eqref{eq:manifold}.  Such solutions are known as spatially homogeneous FLRW solutions.

\begin{remark}[Decelerated expansion] \label{rmk:decel}
Spacetimes with metric of the form \eqref{eq:FLRWgeneral} are said to undergo \emph{accelerated expansion} if the scale factor $a(t)$ satisfies $a''(t) > 0$ for all $t$, \emph{linear expansion} if $a''(t) = 0$ for all $t$, or \emph{decelerated expansion} if $a''(t)<0$ for all $t$.  Note that the solutions \eqref{eq:FLRW1}--\eqref{eq:FLRW2} undergo decelerated expansion.  Previous stability works in cosmological settings have typically considered perturbations of spacetimes undergoing linear or accelerated expansion.  See Section \ref{subsec:previous} below.
\end{remark}

\begin{remark}[Anisotropic spatially homogeneous FLRW solutions] \label{rmk:anisotropic}
	The spatially homogeneous FLRW family of solutions \eqref{eq:FLRW1}--\eqref{eq:FLRW2} lie inside a more general family of \emph{anisotropic} spatially homogeneous FLRW family.  Given any smooth, sufficiently decaying function $F : \mathbb{R}^3 \to [0,\infty)$ such that $F \not\equiv 0$ and
	\[
		\int_{\mathbb{R}^3} F(p) p^i dp = 0, \quad i=1,2,3,	
		\qquad
		\text{and}
		\qquad
		\int_{\mathbb{R}^3} F(p) \frac{p^i p^j}{\vert p \vert} dp = \frac{\varrho}{3} \delta^{ij}, \quad i,j=1,2,3,
	\]
	the metric and particle density
	\[
		g_{\circ} = -dt^2 + a(t)^2 \big( (dx^1)^2 + (dx^2)^2 + (dx^3)^2 \big),
		\qquad
		f_{\circ}(t,x,p)
		=
		F (a(t)^2 p^1, a(t)^2 p^2,a(t)^2 p^3),
	\]
	where
	\[
		a(t) = t^{\frac{1}{2}} \left( \frac{4 \varrho}{3} \right)^{\frac{1}{4}},
		\qquad
		\varrho = \int \vert p \vert F(p) dp,
	\]
	also defines a solution of \eqref{eq:EV1}--\eqref{eq:EV4} on \eqref{eq:manifold}.  Since these solutions are, in general, not spherically symmetric (see Section \ref{subsec:sphericalsymmetry}), consideration here is restricted to the solutions of the form \eqref{eq:FLRW1}--\eqref{eq:FLRW2}.
\end{remark}

\begin{remark}[$\mathbb{T}^3$ spatial topology]
	Each \eqref{eq:FLRW1}--\eqref{eq:FLRW2}, and more generally the solutions of Remark \ref{rmk:anisotropic}, also define a spatially homogeneous FLRW solution of the Einstein--massless Vlasov system on the manifold $(0,\infty)\times \mathbb{T}^3$.  This article will only be concerned with solutions on $(0,\infty)\times \mathbb{R}^3$.  (See also Remark \ref{rmk:T3rates} for a further comment on the solutions on $(0,\infty)\times \mathbb{T}^3$.)
\end{remark}

The components of the energy momentum tensor of \eqref{eq:FLRW1}--\eqref{eq:FLRW2}, with respect to the Cartesian $(t,x^1,x^2,x^3)$ coordinate system, take the form
\begin{equation} \label{eq:EMFLRW}
	(T^{\circ})^{tt}(t,x) = \frac{3}{4t^2},
	\qquad
	(T^{\circ})_i^j(t,x) = \frac{1}{4t^2} \delta_i^j,
	\qquad
	(T^{\circ})_{it}(t,x) = 0,
\end{equation}
for $i,j = 1,2,3$, where $(T^{\circ})_i^j(t,x) = g_{ik}(t,x) (T^{\circ})^{jk}(t,x)$.

The spacetimes $(\mathcal{M}_{\circ},g_{\circ})$ expand from a spacelike singularity at $t=0$ (with Kretchmann scalar $\vert Rm_{\circ} \vert^2_{g_{\circ}} = \frac{3}{2t^4}$) and are future geodesically complete.  See the Penrose diagram of Figure \ref{fig:FLRW}.  
Such solutions feature in the Ehlers--Geren--Sachs Theorem \cite{EGS}, which in particular ensures that all solutions of the Einstein--massless Vlasov system for which $f$ is \emph{isotropic} and \emph{irrotational} are either stationary or described by an FLRW metric as above.

\subsection{First version of the main theorem}

The main result of the present work concerns the future stability of the family \eqref{eq:FLRW1}--\eqref{eq:FLRW2}.

\begin{theorem}[Future stability of FLRW] \label{thm:main}
	Each spatially homogeneous FLRW solution is future nonlinearly stable, as a solution of the Einstein--massless Vlasov system \eqref{eq:EV1}--\eqref{eq:EV4}, to compactly supported spherically symmetric perturbations.
	
	More precisely, consider an FLRW solution of the form \eqref{eq:FLRW1}--\eqref{eq:FLRW2} for some smooth, non-identically vanishing $\mu$ (decaying suitably so that $\varrho$ is defined).  For all compactly supported spherically symmetric initial data sufficiently close to that of \eqref{eq:FLRW1}--\eqref{eq:FLRW2} on $\Sigma_1 = \{t=1\}$, the resulting solution is future geodesically complete, is isometric to \eqref{eq:FLRW1}--\eqref{eq:FLRW2} after some retarded time and, in an appropriately normalised double null gauge, the (appropriately normalised) components of the energy momentum tensor decay to zero as \eqref{eq:EMFLRW} to leading order, and each metric component and Christoffel symbol either remains close or decays to its FLRW value to the past of this retarded time, with quantitative polynomial rates.
\end{theorem}

A more precise version of Theorem \ref{thm:main} is stated below in Section \ref{section:maintheorem}.  See also the Penrose diagram in Figure \ref{fig:FLRW}.

The proof of Theorem \ref{thm:main} is based on an understanding of decay properties of solutions of the massless Vlasov equation.  As such, in addition to Theorem \ref{thm:main}, decay properties of solutions of the massless Vlasov equation on a fixed FLRW background of the form \eqref{eq:FLRW1}--\eqref{eq:FLRW2} are given.  See Theorem \ref{thm:fixedbackground}.

The assumption of compact support can be relaxed, and is made for convenience in order to localise the proof.  Indeed, following \cite{Da}, the perturbation of $f$, in the coupled problem, vanishes close to the centre of spherical symmetry at late times.  The absence of singularities in future evolution therefore follows from the aforementioned quantitative decay properties, along with comparatively soft arguments including an \emph{extension principle around non-central points} (see \cite{DaRe07} or Theorem \ref{thm:noncentext} below).

\begin{remark}[Comparison with decay rates on Minkowski space]
The decay rates of solutions of the massless Vlasov equation on FLRW, obtained in Theorem \ref{thm:fixedbackground}, are later compared with those of the massless Vlasov equation on Minkowski space. When measured with respect to their respective $t$ coordinates, certain components of the energy momentum tensor, with respect to an appropriately normalised double null frame, decay faster in Minkowski space, while others decay faster in FLRW. See Remark \ref{rmk:compareMinkowski}.
\end{remark}

\begin{remark}[Birkhoff-type theorem]
	The proof of Theorem \ref{thm:main} in particular contains a Birkhoff-type theorem for the system \eqref{eq:EV1}--\eqref{eq:EV4} which ensures that any spherically symmetric solution with a regular centre, for which $f$ is equal to its FLRW value $f_{\circ}$ (defined by \eqref{eq:FLRW1}), is locally isometric to the FLRW spacetime $(\mathcal{M}_{\circ},g_{\circ})$.  See the final step in the proof of Theorem \ref{thm:bootstrap} (from which a general theorem can be extracted).  Recall also the Ehlers--Geren--Sachs Theorem \cite{EGS}.
\end{remark}

\begin{remark}[Einstein--Euler with radiation equation of state]
	The metric $g_{\circ}$, defined by the former of \eqref{eq:FLRW1} with $a(t)=t^{\frac{1}{2}}$, also arises as a solution of the Einstein--Euler system, i.\@e.\@ equation \eqref{eq:EV3} with
	\[
		T^{\mu \nu} = (\rho + p) u^{\mu} u^{\nu} + p g^{\mu \nu},
	\]
	for unknowns $g$, four velocity $u^{\mu}$ --- satisfying $g_{\mu \nu} u^{\mu} u^{\nu} = -1$ --- pressure $p$, and density $\rho$, with $p$ and $\rho$ related by the \emph{radiation equation of state}
	\[
		p = \frac{\rho}{3}.
	\]
	The fluid variables for the FLRW metric take the form, in Cartesian coordinates,
	\begin{equation} \label{eq:FLRWfluid}
		u = \partial_t,
		\qquad
		\rho = \frac{3}{4 t^2}.
	\end{equation}
	
	It is known that the solution \eqref{eq:FLRWfluid} is unstable to formation of shock waves for the Euler equations $\nabla_{\mu} T^{\mu \nu} = 0$ on the fixed FLRW background $(\mathcal{M}_{\circ},g_{\circ})$ (see Chapter 9 of \cite{Sp13}).
	
	Note also the linear \emph{Jeans instability} for the Einstein--Euler system linearised around FLRW \cite{Lif} (see also the discussion in \cite{Bo}).
\end{remark}

\subsection{Previous works}
\label{subsec:previous}

There have been previous works on the Einstein--Vlasov system in cosmological settings for Einstein equations with a positive cosmological constant \cite{Ri}, and also for perturbations of the vacuum Milne solution \cite{AnFa,BaFa,Fa}.  Note also related works \cite{HaSp,RoSp,Sp,FaOfWy} on the Einstein--Euler system.  See also \cite{AnMo,FaWy,Four}.
Each of these works considers perturbations of spacetimes with metric of the form \eqref{eq:FLRWgeneral} undergoing \emph{accelerated expansion} or \emph{linear expansion} (see Remark \ref{rmk:decel}).  Contrast with the solutions \eqref{eq:FLRW1}--\eqref{eq:FLRW2} of the present work in which undergo slow \emph{decelerated expansion}.

In the asymptotically flat setting, Minkowski space has been shown to be nonlinearly stable for both the massless and massive Einstein--Vlasov systems, first in spherical symmetry \cite{Da,ReRe}, and later to general perturbations \cite{Tay17,LiTa,FJS}.  The present work is based on the approach of \cite{Da}.  \emph{Anti-de Sitter} space has also been shown to be nonlinearly \emph{unstable} as a solution of the Einstein--massless Vlasov system with a negative cosmological constant \cite{Mosproof}.

Though there is no direct non-relativistic analogue of the present problem, there is an infinite dimensional family of explicit, spatially homogenous, \emph{stationary} solutions of the Vlasov--Poisson system.  The works \cite{HKNR,BeMaMo22,IPWW} on Landau damping, and also the celebrated works \cite{MoVi,BeMaMo,GrNgRo} on the torus $\mathbb{T}^3$, concern the stability of these families of solutions.  Note however that one must apply the \emph{Jeans swindle} to make sense of the problems considered in these works in the gravitational setting.  See also \cite{Kie}.

\subsection{Outline of the paper}
Section \ref{section:EVss} contains preliminaries on the Einstein--massless Vlasov system in double null gauge.  Section \ref{section:FLRW} concerns further properties of the FLRW spacetimes, introduced in Section \ref{subsec:FLRWintro}.  A double null gauge in FLRW is defined in Section \ref{subsec:FLRWdoublenull} and, as a precursor to the proof Theorem \ref{thm:main}, properties of the massless Vlasov equation \eqref{eq:EV1}--\eqref{eq:EV2} on the FLRW background spacetimes are presented in Section \ref{subsec:mvfixedbackground}.  In Section \ref{section:maintheorem} a more precise version of Theorem \ref{thm:main} is formulated, and in Section \ref{section:proof} its proof is given.

\subsection*{Acknowledgements}
I acknowledge support through Royal Society Tata University Research Fellowship URF\textbackslash R1\textbackslash 191409.  
I am grateful to G.\@ Fournodavlos, H.\@ Masaood, and J.\@ Speck for helpful discussions.

\section{The spherically symmetric Einstein--massless Vlasov system}
\label{section:EVss}

This section concerns facts about the Einstein--massless Vlasov system in spherical symmetry.  In Section \ref{subsec:sphericalsymmetry} the spherical symmetry assumption is introduced.  In Section \ref{subsec:doublenull} spherically symmetric double null gauges are introduced, and the residual spherically symmetric double null freedom is discussed.  In Section \ref{subsec:ssEmVdng} the Einstein--massless Vlasov system in a spherically symmetric double null gauge is presented.  In Section \ref{subsec:tandr} functions $t$ and $r$ are introduced.  Finally, in Section \ref{subsec:Cauchyproblem} the Cauchy problem is discussed and an extension principle, which will be used in the proof of Theorem \ref{thm:main}, is stated.

\subsection{Spherical symmetry}
\label{subsec:sphericalsymmetry}

A solution $(\mathcal{M},g,f)$ of \eqref{eq:EV1}--\eqref{eq:EV4} is \emph{spherically symmetric} if $SO(3)$ acts by isometry on $(\mathcal{M},g)$ and preserves $f$.  More precisely, it is assumed that there is a smooth isometric action $\mathcal{O} \colon \mathcal{M} \times SO(3) \to \mathcal{M}$ on $(\mathcal{M},g)$ such that, for each $p\in \mathcal{M}$, the orbit of the action satisfies either $Orb(p) \simeq S^2$, or $Orb(p) = \{p\}$ and, moreover, for each of the generators $\Omega_1, \Omega_2, \Omega_3$ of the $SO(3)$ action, the corresponding flows $\Phi^i\colon (0,2\pi) \times \mathcal{M} \to \mathcal{M}$ satisfy
\begin{equation} \label{eq:ssf}
	f(x,p) = f(\Phi^i_s(x),(\Phi^i_s)_* p), \qquad \text{for all } s\in (0,2\pi), \quad i = 1,2,3,
\end{equation}
for all $(x,p) \in \mathcal{P}$.  It is assumed that the centre
\[
	\Gamma = \{ x \in \mathcal{M} \mid \mathcal{O}(x,\mathfrak{s}) = x \text{ for all } \mathfrak{s} \in SO(3)\},
\]
consists of a single timelike geodesic, and that $\mathcal{M}\smallsetminus \Gamma$ splits diffeomorphically as
\[
	\mathcal{M}\smallsetminus \Gamma \simeq \mathcal{Q} \times S^2,
\]
for a smooth 2-manifold $\mathcal{Q}= (\mathcal{M} \smallsetminus \Gamma) /SO(3)$.
The quotient space $\mathcal{M}/SO(3)$, also denoted $\mathcal{Q}$, is viewed a manifold with boundary, with boundary identified with the centre $\Gamma$.

\subsection{Double null gauge} \label{subsec:doublenull}
A double null gauge consists of functions $u,v\colon \mathcal{Q} \to \mathbb{R}$ such that the level hypersurfaces of $u$ foliate $\mathcal{Q}$ by outgoing lines which are null with respect to the induced metric on $\mathcal{Q}$, and the level hypersurfaces of $v$ foliate $\mathcal{Q}$ by ingoing null lines.  The level hypersurfaces of $u$ and $v$ lift to outgoing and incoming null cones of $\mathcal{M}$ and such a double null gauge can be complemented with local coordinates $(\theta^1,\theta^2)$ on $S^2$ to local coordinates $(u,v,\theta^1,\theta^2)$ for $\mathcal{M}$.

For a given double null gauge, the metric $g$ can be written in double null form
\begin{equation} \label{eq:gdoublenull}
	g = - \Omega^2 du dv + R^2 \gamma,
\end{equation}
where $\gamma$ is the unit round metric on $S^2$, where $\Omega$ is a function on $\mathcal{Q}$ and $R \colon \mathcal{Q} \to \mathbb{R}$ is the area radius function
\[
	R(u,v) = \sqrt{\mathrm{Area}(S_{u,v})/4\pi},
\]
where $S_{u,v} = \{ (u,v) \} \times S^2 \subset \mathcal{M}$.
The area radius $R$ extends regularly to $0$ on the centre $\Gamma$.  Define
\[
	\lambda = \partial_v R,
	\qquad
	\nu = \partial_u R.
\]

\begin{remark}[Residual gauge freedom] \label{rmk:residualgauge}
	There is residual double null gauge freedom present in \eqref{eq:gdoublenull}.  If $\widetilde{u}, \widetilde{v} \colon I \to \mathbb{R}$ are increasing functions, with $I \subset \mathbb{R}$ a suitable interval, then the change $u \mapsto \widetilde{u}(u)$, $v \mapsto \widetilde{v}(v)$ preserves the double null form \eqref{eq:gdoublenull}.  Under such a change, the metric takes the form
	\[
		g
		=
		- \widetilde{\Omega}^2(\widetilde{u},\widetilde{v}) d \widetilde{u} d \widetilde{v} + R(\widetilde{u},\widetilde{v})^2 \gamma
		=
		- \widetilde{\Omega}^2(\widetilde{u}(u),\widetilde{v}(v)) \frac{d\widetilde{u}}{du} (u) \frac{d\widetilde{v}}{dv}(v) d u d v + R(\widetilde{u}(u),\widetilde{v}(v))^2 \gamma.
	\]
	It thus follows that
	\begin{align}
		\Omega^2(u,v) 
		&
		=
		\widetilde{\Omega}^2(\widetilde{u}(u),\widetilde{v}(v)) \widetilde{u}' (u) \widetilde{v}' (v),
		\label{eq:gauge1}
		\\
		R(u,v) 
		&
		=
		\widetilde{R}(\widetilde{u}(u),\widetilde{v}(v)),
		\label{eq:gauge2}
		\\
		\nu(u,v) 
		&
		=
		\widetilde{\nu}(\widetilde{u}(u),\widetilde{v}(v)) \widetilde{u}' (u),
		\label{eq:gauge3}
		\\
		\lambda(u,v) 
		&
		=
		\widetilde{\lambda}(\widetilde{u}(u),\widetilde{v}(v)) \widetilde{v}' (v).
		\label{eq:gauge4}
	\end{align}
\end{remark}

The inverse metric of \eqref{eq:gdoublenull} takes the form
\[
	g^{-1} = -2 \Omega^{-2} (\partial_u \otimes \partial_v + \partial_v \otimes \partial_u) + R^{-2} \gamma^{AB} \partial_{\theta^A} \otimes \partial_{\theta^B},
\]
where $A$ and $B$ range over $1$ and $2$.  The nonvanishing Christoffel symbols of $g$ are
\[
	\Gamma^u_{AB} = \frac{2 \lambda R}{ \Omega^2} \gamma_{AB},
	\quad
	\Gamma^v_{AB} = \frac{2 \nu R}{\Omega^2} \gamma_{AB},
	\quad
	\Gamma^A_{Bv} = \frac{\lambda}{R} \delta^A_B,
	\quad
	\Gamma^A_{Bu} = \frac{\nu}{R} \delta^A_B,
	\quad
	\Gamma^u_{uu} = \partial_u \log \Omega^2,
	\quad
	\Gamma^v_{vv} = \partial_u \log \Omega^2,
\]
and
\[
	\Gamma^A_{BC} = \Gammaslash^A_{BC},
\]
where $\Gammaslash^A_{BC}$ are the Christoffel symbols of $(S^2,\gamma)$.

For a given spherically symmetric double null gauge $(u,v,\theta^1,\theta^2)$ for $(\mathcal{M},g)$, one defines a corresponding conjugate coordinate system $(u,v,\theta^1,\theta^2,p^u,p^v,p^1,p^2)$ for $T\mathcal{M}$, whereby $(u,v,\theta^1,\theta^2,p^u,p^v,p^1,p^2)$ defines the point
\[
	\big(p^{u} \partial_u + p^v \partial_v + p^{A} \partial_{\theta^A} \big) \big\vert_{(u,v,\theta^1,\theta^2)} \in T\mathcal{M}.
\]
Such a coordinate system induces a coordinate system $(u,v,\theta^1,\theta^2,p^v,p^1,p^2)$ on the mass shell $P$, with $p^u$ defined by the mass shell relation \eqref{eq:EV2}, which takes the form
\begin{equation} \label{eq:defofpu}
	p^u
	=
	\frac{R^2 \gamma_{AB} p^A p^B}{\Omega^2 p^v}
	=
	\frac{L^2}{\Omega^2 R^2 p^v},
\end{equation}
where the angular momentum $L \colon P \to [0,\infty)$ is defined by
\begin{equation} \label{eq:am}
	L(x,p)^2 = R^4 \gamma_{AB} p^A p^B.
\end{equation}

The volume form on
\[
	P_x = \{ p \in T_x\mathcal{M} \mid g_{x} (p, p) = 0 \} \subset T_x\mathcal{M},
\]
in this induced coordinate system takes the form
\[
	\frac{R^2 \sqrt{\det \gamma}}{p^v} dp^v \wedge dp^1 \wedge dp^2.
\]

The condition \eqref{eq:ssf} that $f$ is spherically symmetric means that, in a given double null gauge, $f$ can be written as a function --- which, abusing notation slightly, is also denoted $f\colon \mathcal{Q} \times [0,\infty) \times [0,\infty) \to [0,\infty)$ --- of $u$, $v$, $p^v$, and $L = (R^4 \gamma_{AB} p^A p^B)^{\frac{1}{2}}$,
\begin{equation} \label{eq:reducedf}
	f(x,p) = f(u,v,p^v,L).
\end{equation}

The energy momentum tensor $T$ on $\mathcal{M}$ takes the form
\[
	T = T_{uu}(u,v) du^2 + 2 T_{uv}(u,v) du dv + T_{vv}(u,v) dv^2 + T_{AB}(u,v) d\theta^A d\theta^B.
\]
In double null gauge, the components take the form
\begin{align}
	T_{uu}(u,v)
	&
	=
	\frac{\Omega^4}{R^2} \frac{\pi}{2} \int_{0}^{\infty} \int_0^{\infty}
	f(u,v,p^v,L) 
	\, p^v
	\, L \, dL \,
	dp^v,
	\label{eq:emtensordoublenull1}
	\\
	T_{uv}(u,v)
	&
	=
	\frac{\Omega^4}{R^2} \frac{\pi}{2} \int_{0}^{\infty} \int_0^{\infty}
	f(u,v,p^v,L) 
	\, p^u
	\, L \, dL \,
	dp^v,
	\label{eq:emtensordoublenull2}
	\\
	T_{vv}(u,v)
	&
	=
	\frac{\Omega^4}{R^2} \frac{\pi}{2} \int_{0}^{\infty} \int_0^{\infty}
	f(u,v,p^v,L) 
	\, \frac{(p^u)^2}{p^v}
	\, L \, dL \,
	dp^v,
	\label{eq:emtensordoublenull3}
\end{align}
with $p^u = p^u(u,v,p^v,L)$ defined by \eqref{eq:defofpu}.  It moreover follows from the mass shell relation \eqref{eq:defofpu} that
\[
	\gamma^{AB} T_{AB} = \frac{R^2}{4\Omega^2} T_{uv}.
\]

\subsection{The spherically symmetric Einstein--massless Vlasov system in double null gauge}
\label{subsec:ssEmVdng}

The Einstein equations \eqref{eq:EV3} in spherical symmetry take the form of the following system of equations for $(\Omega, R,f)$,
\begin{align}
	\partial_u \partial_v (R^2)
	&
	=
	-\frac{\Omega^2}{2} + R^2 T_{uv},
	\label{eq:Einsteinss1}
	\\
	\partial_u \partial_v \log \Omega^2
	&
	=
	\frac{\Omega^2}{2 R^2} \big( 1 + 4 \Omega^{-2} \partial_u R \partial_v R \big)
	-
	T_{uv}
	-
	\frac{\Omega^2}{4R^2} \gamma^{AB} T_{AB},
	\label{eq:Einsteinss2}
	\\
	\partial_u(\Omega^{-2} \partial_u R)
	&
	=
	- \frac{1}{2} R \Omega^{-2} T_{uu},
	\label{eq:Einsteinss3}
	\\
	\partial_v(\Omega^{-2} \partial_v R)
	&
	=
	- \frac{1}{2} R \Omega^{-2} T_{vv},
	\label{eq:Einsteinss4}
\end{align}
where $\gslash = R^2 \gamma$.  The mass shell relation \eqref{eq:EV2} takes the form
\begin{equation} \label{eq:massshell}
	- \Omega^2 p^u p^v + R^2 \gamma_{AB} p^A p^B = 0.
\end{equation}
In particular, $\Omega^2 T^{uv} = R^2 \gamma_{AB} T^{AB}$, or $4 \Omega^{-2} T_{uv} = R^{-2} \gamma^{AB} T_{AB}$, and so equation \eqref{eq:Einsteinss2} can be rewritten
\begin{align}
	\partial_u \partial_v \log \Omega^2
	&
	=
	\frac{\Omega^2}{2 R^2} \big( 1 + 4 \Omega^{-2} \partial_u R \partial_v R \big)
	-
	2 T_{uv},
	\label{eq:Einsteinss2var}
\end{align}
The Vlasov equation \eqref{eq:EV1} in spherical symmetry takes the form
\begin{align} \label{eq:Vlasovss}
	X(f)
	\equiv
	p^u \partial_u f
	+
	p^v \partial_v f
	-
	\Big(R^2 \gamma_{AB} p^A p^B \frac{2 \nu}{R \Omega^2}
	+
	(p^v)^2 \partial_v \log \Omega^2 \Big) \partial_{p^v} f
	-
	2
	\frac{\lambda p^v + \nu p^u}{R} L \partial_{L} f
	=
	0.
\end{align}

The angular momentum $L \colon P \to [0,\infty)$ defined by \eqref{eq:am} is preserved by the geodesic flow
\begin{equation} \label{eq:consofam}
	X(L) = 0.
\end{equation}

The system \eqref{eq:Einsteinss1}--\eqref{eq:Vlasovss} is equivalent to the Einstein--massless Vlasov system \eqref{eq:EV1}--\eqref{eq:EV4} in the sense that any spherically symmetric solution $(\mathcal{M},g,f)$ of \eqref{eq:Einsteinss1}--\eqref{eq:Vlasovss} defines a solution $(\mathcal{Q},\Omega^2,R,f)$ of the reduced system \eqref{eq:Einsteinss1}--\eqref{eq:Vlasovss} and, conversely, any solution $(\mathcal{Q},\Omega^2,R,f)$ of the reduced system \eqref{eq:Einsteinss1}--\eqref{eq:Vlasovss} defines a solution of \eqref{eq:EV1}--\eqref{eq:EV4} with $\mathcal{M} = \mathcal{Q} \times S^2$ and $g$ defined by \eqref{eq:gdoublenull}.

See \cite{Mowp} for more on the spherically symmetric Einstein--massless Vlasov system.

\subsection{The functions $t$ and $r$} \label{subsec:tandr}
Define functions $t$ and $r$ by
\[
	t(u,v) = \frac{1}{4}(v+u)^2, \qquad r(u,v) = v - u.
\]
When the FLRW spacetime \eqref{eq:FLRW1}--\eqref{eq:FLRW2} is expressed in the double null gauge introduced in Section \ref{subsec:FLRWdoublenull} below, the functions $t$ and $r$ coincide with those of Section \ref{subsec:FLRWintro}.

\subsection{The Cauchy problem}
\label{subsec:Cauchyproblem}

The Einstein--massless Vlasov system \eqref{eq:EV1}--\eqref{eq:EV4} admits the following local existence theorem.  Initial data consists of a Riemannian $3$ manifold $(\Sigma,\overline{g})$, along with a symmetric $(0,2)$ tensor $K$ on $\Sigma$, and a function $f_1 \colon T\Sigma \to [0,\infty)$, satisfying constraint equations
\begin{align*}
	\overline{R} - \vert K \vert_{\overline{g}}^2 + (\overline{\tr} K)^2
	=
	2 \sqrt{\det \overline{g}} \int f_1 \vert p \vert_{\overline{g}} dp^1 dp^2 dp^3,
	\qquad
	\overline{\mathrm{div}} K^i - \nabla^i \overline{\tr} K
	=
	\sqrt{\det \overline{g}} \int f_1 p^i dp^1 dp^2 dp^3,
\end{align*}
where $\overline{R}$ is the scalar curvature and $\overline{\nabla}$ is the Levi-Civita connection of $\overline{g}$.  A development of initial data $(\Sigma, \overline{g},K,f_1)$ is a globally hyperbolic spacetime $(\mathcal{M},g)$ which admits $\Sigma$ as a Cauchy hypersurface, with induced first and second fundamental form $\overline{g}$ and $K$ respectively, such that the restriction of $f$ to the mass shell over $\Sigma$, $P\vert_{\Sigma}$, coincides with $f_1$, when $T \Sigma$ is appropriately identified with $P\vert_{\Sigma}$.  See \cite{Ri} for more on the Cauchy problem for the Einstein--Vlasov system.

\begin{theorem}[Local well posedness of the Cauchy problem for the Einstein--massless Vlasov system \cite{ChBr71,ChBrGe}] \label{thm:Cauchyproblem}
	For any smooth initial data set $(\Sigma, \overline{g},K,f_1)$ for the Einstein--massless Vlasov system, as above, there exists a unique maximal development $(\mathcal{M},g,f)$ solving the system \eqref{eq:EV1}--\eqref{eq:EV4}.  Moreover, if $(\Sigma, \overline{g},K,f_1)$ is spherically symmetric then the maximal development $(\mathcal{M},g,f)$ is spherically symmetric.
\end{theorem}

See also \cite{Mowp} for a version of Theorem \ref{thm:Cauchyproblem} in a more general spherically symmetric setting.

The following extension principle for spherically symmetric solutions of the system \eqref{eq:EV1}--\eqref{eq:EV4}, concerning \emph{non-central points}, will be used in what follows.  It is assumed that $(\mathcal{M},g,f)$ is the maximal development of a smooth spherically symmetric initial data set $(\Sigma, \overline{g},K,f_1)$ for the Einstein--massless Vlasov system \eqref{eq:EV1}--\eqref{eq:EV4}, and that $\mathcal{Q}$ is the quotient manifold of $\mathcal{M}$ by the action of the $SO(3)$ isometry, as in Section \ref{subsec:sphericalsymmetry} above.  

\begin{theorem}[Extension principle, around \emph{non-central points}, for the spherically symmetric Einstein--massless Vlasov system \cite{DaRe07}] \label{thm:noncentext}
	Let $\mathcal{Q}$ be as above.  Let $(u_*,v_*)$ be such that there exists $U<u_*$ and $V <v_*$ such that the characteristic diamond $\mathcal{D}_{U,V}^{u_*,v_*} = \{ U \leq u < u_*, V \leq v < v_*\}$ is contained in $\mathcal{Q}$,
	\begin{equation} \label{eq:noncentral}
		\mathcal{D}_{U,V}^{u_*,v_*} \subset \mathcal{Q}.
	\end{equation}
	Assume also that there exists $0 < R_0 < R_1$ such that
	\[
		\int^{u_*}_U\int^{v_*}_V 
		\Omega^2(u,v) du dv < \infty,
		\qquad
		\text{and}
		\qquad
		R_0 \leq R(u,v) \leq R_1,
		\quad
		\text{for all } (u,v) \in \mathcal{D}_{U,V}^{u_*,v_*},
	\]
	and that $f(u,v,\cdot,\cdot)$ is compactly supported for all $(u,v) \in \mathcal{D}_{U,V}^{u_*,v_*}$.  Then $(u_*,v_*) \in \mathcal{Q}$.
\end{theorem}

For a proof of Theorem \ref{thm:noncentext} see Section 4.\@3 of \cite{DaRe07} (where one has to replace the massive mass shell relation with its massless analogue \eqref{eq:massshell}, which does not affect the proof).  Theorem \ref{thm:noncentext} is preferred to the softer extension principle of \cite{DaRe} in view of the presence of the anti-trapped surfaces in the spacetimes under consideration (see Remark \ref{rmk:antitrapped} below).

\section{The FLRW spacetimes}
\label{section:FLRW}

Recall the FLRW spacetimes introduced in Section \ref{subsec:FLRWintro}.  In Section \ref{subsec:FLRWdoublenull} a double null gauge is introduced in each of these spacetimes.  Section \ref{subsec:mvfixedbackground} concerns properties of the massless Vlasov equation on these spacetimes.  Though the main theorem of Section \ref{subsec:mvfixedbackground}, Theorem \ref{thm:fixedbackground}, is not used in the proof of Theorem \ref{thm:main}, its proof is presented as a simple precursor to that of Theorem \ref{thm:main}.  There is no symmetry assumption required, however, for Theorem \ref{thm:fixedbackground}.

\subsection{The FLRW metrics in double null gauge}
\label{subsec:FLRWdoublenull}

Recall that the FLRW spacetime $(\mathcal{M}_{\circ},g_{\circ})$ takes the form
\[
	\mathcal{M}_{\circ} = (0,\infty) \times \mathbb{R}^3,
	\qquad
	g_{\circ} = -dt^2 + t \left( \frac{4\varrho}{3} \right)^{\frac{1}{2}} \big( (dx^1)^2 + (dx^2)^2 + (dx^3)^2 \big).
\]
For $t\in (0,\infty)$, define
\[
	\Sigma_t = \{t \} \times \mathbb{R}^3 \subset \mathcal{M}_{\circ}.
\]

Define double null coordinates
\[
	u = t^{\frac{1}{2}} - \frac{r}{2},
	\qquad
	v = t^{\frac{1}{2}} + \frac{r}{2},
	\qquad
	\text{where }
	r = \left( \frac{4\varrho}{3} \right)^{\frac{1}{4}} \sqrt{(x^1)^2 + (x^2)^2 + (x^3)^2},
\]
and note that, since $t\geq 0$ and $r\geq 0$ in $\mathcal{M}_{\circ}$,
\[
	v \geq 0, \qquad v \geq u, \qquad v \geq - u, \qquad \text{in } \mathcal{M}_{\circ}.
\]

The FLRW metric $g_{\circ}$ in the above double null gauge takes the form
\begin{equation} \label{eq:FLRWdoublenull}
	g_{\circ} = -4t du dv + tr^2 \gamma,
	\qquad
	t = \frac{1}{4}(v+u)^2,
	\quad
	r = v-u,
\end{equation}
defined on the quotient manifold
\[
	\mathcal{Q}_{\circ} = \{ (u,v) \in \mathbb{R}^2 \mid v \geq 0, v \geq u, v \geq - u \}.
\]

The metric $g_{\circ}$ can be written
\[
	g_{\circ} = -\Omega_{\circ}^2 du dv + R_{\circ}^2 \gamma,
	\qquad
	\text{where}
	\qquad
	\Omega^2_{\circ} = 4t,
	\qquad
	R_{\circ} = t^{\frac{1}{2}} r,
\]
and moreover
\[
	\lambda_{\circ} = \partial_v R_{\circ} = v = \frac{r}{2} + t^{\frac{1}{2}},
	\qquad
	\nu_{\circ} = \partial_u R_{\circ} = - u = \frac{r}{2} - t^{\frac{1}{2}},
	\qquad
	\sqrt{-\det g_{\circ}} = 2t^2r^2 \sqrt{\det \gamma}.
\]

\begin{remark}[Anti-trapped spheres] \label{rmk:antitrapped}
	Note that the FLRW spacetime $(\mathcal{M}_{\circ},g_{\circ})$ contains \emph{anti-trapped} spheres of symmetry: if $(u,v)\in \mathcal{Q}_{\circ}$ is such that $u <0$, then $\lambda_{\circ}(u,v) > 0$ and $\nu_{\circ}(u,v) >0$.
\end{remark}

The non-vanishing Christoffel symbols of \eqref{eq:FLRWdoublenull} take the form
\[	
	\mathring{\Gamma}^u_{AB} = \frac{rv}{2t^{\frac{1}{2}}} \gamma_{AB} = \frac{1}{2} \left( \frac{1}{2t^{\frac{1}{2}}} + \frac{1}{r} \right) r^2 \gamma_{AB},
	\qquad
	\mathring{\Gamma}^v_{AB} = - \frac{ru}{2t^{\frac{1}{2}}} \gamma_{AB} = \frac{1}{2} \left( \frac{1}{2t^{\frac{1}{2}}} - \frac{1}{r} \right) r^2 \gamma_{AB},
\]
\[
	\mathring{\Gamma}_{uu}^u = \mathring{\Gamma}_{vv}^v = t^{-\frac{1}{2}},
	\qquad
	\mathring{\Gamma}^A_{Bv} = \left( \frac{1}{2t^{\frac{1}{2}}} + \frac{1}{r} \right) \delta_A^B,
	\qquad
	\mathring{\Gamma}^A_{Bu} = \left( \frac{1}{2t^{\frac{1}{2}}} - \frac{1}{r} \right) \delta_A^B,
	\qquad
	\mathring{\Gamma}^A_{BC} = \Gammaslash^A_{BC},
\]
where $\Gammaslash^A_{BC}$ are the Christoffel symbols of $(S^2,\gamma)$.

It follows from \eqref{eq:EMFLRW} that the null components of the energy momentum tensor of $f_{\circ}$ satisfy
\begin{equation} \label{eq:emFLRWdoublenull}
	T^{\circ}_{uu} = \frac{1}{t},
	\qquad
	T^{\circ}_{uv} = \frac{1}{2t},
	\qquad
	T^{\circ}_{vv} = \frac{1}{t}.
\end{equation}
The mass shell relation \eqref{eq:massshell} takes the form
\begin{equation} \label{eq:massshellFLRW}
	-4p^u p^v + r^2 \gamma_{AB} p^A p^B = 0.
\end{equation}
Recall the expression \eqref{eq:defofpu} for $p^u$, and define
\[
	\mathring{p}^u = \frac{r^2 \gamma_{AB} p^A p^B}{4p^v}.
\]
Where there is no ambiguity, for example in Section \ref{subsec:mvfixedbackground}, $\mathring{p}^u$ will also be denoted $p^u$.

Given a solution $(\mathcal{Q},\Omega^2,R,f)$ of the spherically symmetric Einstein--massless Vlasov system in double null gauge \eqref{eq:Einsteinss1}--\eqref{eq:Vlasovss}, one identifies with the spatially homogeneous FLRW solutions \eqref{eq:FLRW1}--\eqref{eq:FLRW2} using values of the coordinates $(u,v)$ coordinates for $\mathcal{Q}$, and the values of the coordinates $(u,v,p^v,L)$ for $\mathcal{Q} \times [0,\infty) \times [0,\infty)$, along with the above double null gauge for the spatially homogeneous FLRW solutions.  For example 
\[
	(\Omega^2 - \Omega^2_{\circ})(u,v) = \Omega^2(u,v) - (v+u)^2.
\]
Note in particular that, if $f = f_{\circ}$, then
\[
	\Omega^{-4} R^{-2} T_{uu} = \Omega^{-4}_{\circ} R^{-2}_{\circ} T^{\circ}_{uu},
	\qquad
	\Omega^{-2} R^{-4} T_{uv} = \Omega^{-2}_{\circ} R^{-4}_{\circ} T^{\circ}_{uv},
	\qquad
	R^{-6} T_{vv} = R^{-6}_{\circ} T^{\circ}_{vv}.
\]
The asymmetry in $u$ and $v$ arises from the choice of parameterising the mass shell by $p^v$, rather than $p^u$.

\subsection{The massless Vlasov equation on an FLRW background}
\label{subsec:mvfixedbackground}

The proof of Theorem \ref{thm:main} is based on the following proof of decay of components of solutions of the massless Vlasov equation on a fixed FLRW backround.

\begin{theorem}[Solutions of massless Vlasov equation on a fixed FLRW background with $\mathbb{R}^3$ spatial topology] \label{thm:fixedbackground}
	Let $f$ be a solution of the massless Vlasov equation \eqref{eq:EV1} on $((0,\infty)\times \mathbb{R}^3,g_{\circ})$, where $g_{\circ}$ is the FLRW metric \eqref{eq:FLRWdoublenull}, such that $f_1 = f\vert_{\{t=1\}}$ is compactly supported.  The components of the energy momentum tensor satisfy
	\begin{equation} \label{eq:EMtensoronFLRW}
		T_{uu} \lesssim \frac{\Vert f_1 \Vert_{L^{\infty}}}{t^2},
		\qquad
		T_{uv} \lesssim \frac{\Vert f_1 \Vert_{L^{\infty}}}{t^3},
		\qquad
		T_{vv} \lesssim \frac{\Vert f_1 \Vert_{L^{\infty}}}{t^4},
	\end{equation}
	for $t\geq 1$.  Moreover, there exist $U_0 \leq U_1$ such that
	\begin{equation} \label{eq:FLRWsuppf}
		\pi(\supp(f)) \subset \{ U_0 \leq u \leq U_1 \},
	\end{equation}
	where $\pi \colon P \to \mathcal{M}$ denotes the natural projection.
\end{theorem}

The support of such solutions are depicted in Figure \ref{fig:FLRW}.

\begin{figure}
\centering{
\includegraphics[scale=0.4]{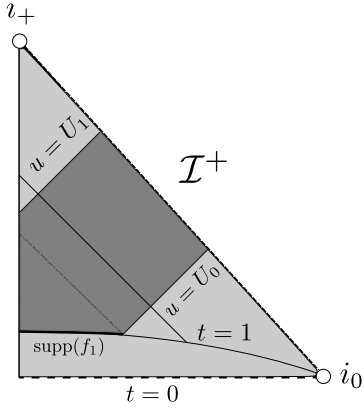}
\caption{Penrose diagram of the FLRW spacetime.  Solutions of the massless Vlasov equation arising from compactly supported initial data on $\{t=1\}$ are supported in the darker shaded region.  The solutions, of the coupled system, of Theorem \ref{thm:main} admit a similar Penrose diagram, where now the perturbation is supported in the darker shaded region.} \label{fig:FLRW}
}
\end{figure}

The proof of Theorem \ref{thm:fixedbackground} is based on properties of null geodesics in FLRW.  Given a null geodesic $\gamma\colon [1,\infty) \to \mathcal{M}$, write the tangent vector to $\gamma$ as
\[
	\dot{\gamma}(s) = p^u(s) \partial_u + p^v(s) \partial_v + p^A (s) \partial_{\theta^A}.
\]
The geodesic equations for $p^u(s)$ and $p^v(s)$ take the form
\begin{align}
	\dot{p}^u(s)
	+
	t^{-\frac{1}{2}} \big( p^u(s)\big)^2
	+
	\frac{1}{2} \left( \frac{1}{2t^{\frac{1}{2}}} + \frac{1}{r} \right) r^2 \gamma_{AB} p^A p^B
	&
	=
	0,
	\label{eq:FLRWgeo1}
	\\
	\dot{p}^v(s)
	+
	t^{-\frac{1}{2}} \big( p^v(s)\big)^2
	+
	\frac{1}{2} \left( \frac{1}{2t^{\frac{1}{2}}} - \frac{1}{r} \right) r^2 \gamma_{AB} p^A p^B
	&
	=
	0.
	\label{eq:FLRWgeo2}
\end{align}
One has equations for $t(s) = t(\gamma(s))$ and $r(s) = r(\gamma(s))$ along the geodesic
\begin{equation} \label{eq:FLRWtdotrdot}
	\dot{t}(s) = t(s)^{\frac{1}{2}} \big( p^v(s) + p^u(s) \big),
	\qquad
	\dot{r}(s) = p^v(s) - p^u(s).
\end{equation}
Note in particular that
\begin{equation} \label{eq:FLRWpr}
	\frac{d}{ds} \big( t(p^v - p^u) \big) = \frac{t}{r} r^2 \gamma_{AB} p^A p^B \geq 0.
\end{equation}

\begin{proposition}[Null geodesics in FLRW] \label{prop:geodesics}
	Let $\mathcal{B} \subset P\vert_{\Sigma_1}$ be a compact subset of the mass shell over $\Sigma_1$ and let $\gamma \colon [1,\infty) \to \mathcal{M}$ be a future directed null geodesic such that $(\gamma(1), \dot{\gamma}(1)) \in \mathcal{B}$ and let $v_0$ be sufficiently large.  There exist $0<c<1<C$, $L_0 \geq 0$, and $s_0 \in [1,\infty)$ such that $v(\gamma(s_0)) = v_0$ and the components of the tangent vector to $\gamma$ satisfy, for all $s\geq s_0$,
	\begin{equation} \label{eq:FLRWcomponents}
		\left( \gamma_{AB} p^A(s) p^B(s) \right)^{\frac{1}{2}} \leq \frac{L_0}{t(s) r(s)^2},
		\qquad
		0 \leq p^u(s) \leq \frac{Cp^v(s_0)}{t(s) r(s)^2},
		\qquad
		c \frac{p^v(s_0)}{t(s)} \leq p^v(s) \leq C \frac{p^v(s_0)}{t(s)}.
	\end{equation}
	Moreover, there exists retarded times $U_0 <0<U_1$ such that the $u$ component of $\gamma$ satisfies
	\begin{equation} \label{eq:FLRWubound}
		U_0 \leq u(s) \leq U_1,
		\qquad
		\text{for all }
		s \geq 1.
	\end{equation}
\end{proposition}

\begin{proof}
	Provided $v_0$ is suitably large, the existence of $s_0 \in [1,\infty)$ such that $v(\gamma(s_0)) = v_0$ follows from the compactness of $\mathcal{B}$.
	
	Consider now the bounds \eqref{eq:FLRWcomponents}.  For simplicity, suppose that $r(\gamma(1)) \neq 0$ (otherwise replace $1$ in the following with $1+\epsilon$ for some $\epsilon>0$).
	By the conservation of angular momentum \eqref{eq:consofam}, one has the conservation law
	\begin{equation} \label{eq:FLRWconsam}
		L(s)^2
		:=
		t(s)^2 r(s)^4 \gamma_{AB} p^A(s) p^B(s)
		=
		t(1)^2 r(1)^4 \gamma_{AB} p^A(1) p^B(1)
		\qquad
		\text{for all }
		s \geq 1,
	\end{equation}
	from which the first of \eqref{eq:FLRWcomponents} trivially follows.

	For the remaining inequalities of \eqref{eq:FLRWcomponents}, suppose first that $\gamma$ is radial, so that the conserved quantity $L(s)$ vanishes.  It follows from the mass shell relation \eqref{eq:massshellFLRW} that either $p^u(1) = 0$ or $p^v(1) = 0$.  The geodesic equations \eqref{eq:FLRWgeo1}--\eqref{eq:FLRWgeo2} take the form
	\[
		\frac{d}{ds} \left( t p^u \right) = 0,
		\qquad
		\frac{d}{ds} \left( t p^v \right) = 0,
	\]
	and so, as long as the geodesic remains away from the centre where the coordinate system breaks down, one has $t(s) p^u(s) = p^u(1)$ and $t(s) p^v(s) = p^v(1)$.  If $p^u(1) = 0$ then, by \eqref{eq:FLRWtdotrdot}, $\dot{r}(s) >0$ for all $s\geq 1$ and so the solution remains away from the centre.  The bounds \eqref{eq:FLRWcomponents} then trivially hold.  If $p^v(1)=0$ then equations \eqref{eq:FLRWtdotrdot} imply
	\[
		t(s) = \Big( \frac{3}{2} (s-1)p^u(1) + 1 \Big)^{\frac{2}{3}},
		\qquad
		r(s) = r(1) - \int_1^s p^u(1) \Big( \frac{3}{2} (s'-1)p^u(1) + 1 \Big)^{-\frac{2}{3}} ds'.
	\]
	It follows that the there is a time $s_* \geq 1$ at which the geodesic hits the centre, i.\@e.\@ a time at which $r(s_*) = 0$.  The geodesic then becomes outgoing, with
	\[
		p^u(s) = 0 \text{ for } s> s_*, \qquad \text{and} \qquad \lim_{s \downarrow s_*} p^v(s) = \lim_{s \uparrow s_*} p^u(s) = \frac{p^u(1)}{t(s_*)}.
	\]
	The bounds \eqref{eq:FLRWcomponents} then again trivially hold, as before.
	
	Suppose now that $\gamma$ is not radial, so that $L\neq 0$, where $L$ is defined by \eqref{eq:FLRWconsam}.  Recall (see \eqref{eq:FLRWpr}) that
	\[
		\frac{d}{ds} \big( t(p^v-p^u) \big) = \frac{L^2}{t r^3} \geq 0.
	\]
	If $p^v(1) \geq p^u(1)$ then it follows that $p^v(s) \geq p^u(s)$ for all $s \geq 1$ and, by the latter of \eqref{eq:FLRWtdotrdot}, $r(s)$ is non-decreasing in $s$ and the geodesic remains away from the centre.  If $p^v(1) < p^u(1)$ then, since $\int_0^{\epsilon} r^{-3} dr = \infty$ for all $\epsilon >0$, it follows that there exists a time $s_* \geq 1$ such that $p^v(s_*) = p^u(s_*)$, and moreover $r(s) >0$ for all $1 \leq s \leq s_*$.  As above $r(s)$ is then non-decreasing for $s \geq s_*$ and so the geodesic again remains away from the centre.  Using now the mass shell relation \eqref{eq:massshellFLRW} and the equations \eqref{eq:FLRWgeo1}--\eqref{eq:FLRWgeo2},
	\[
		\frac{d}{ds} \left( \frac{r^2 p^u}{p^v} \right)
		=
		-4r \frac{(p^u)^2}{p^v}
		\leq
		0,
	\]
	and so (defining $s_*=1$ if $p^v(1) \geq p^u(1)$ and integrating from $s=s_*$),
	\begin{equation} \label{eq:FLRWr2pupv}
		r(s)^2 p^u(s) \leq r(s_*)^2 p^v(s), \qquad \text{for all } s \geq s_*.
	\end{equation}
	
	Consider the latter of \eqref{eq:FLRWcomponents}.  It follows from \eqref{eq:FLRWgeo2} and \eqref{eq:FLRWtdotrdot} that
	\begin{align}
		\frac{d}{ds} \left( t p^v \right)
		=
		\frac{t}{2r} r^2 \gamma_{AB} p^A p^B
		=
		\frac{2}{t r^3} \frac{t^2 r^4 \gamma_{AB} p^A p^B}{4}
		=
		\frac{2}{t(s)r(s)^3} t(s_0)^2r(s_0)^2 p^u(s_0) p^v(s_0)
		,
		\label{eq:FLRWpv}
	\end{align}
	The lower bound follows from the sign of the right hand side of \eqref{eq:FLRWpv}.  For the upper bound, \eqref{eq:FLRWr2pupv} guarantees that, if $v_0$ is sufficiently large (and hence $s_0 \geq s_*$ is sufficiently large), then $p^v(s_0) \geq 2 p^u(s_0)$.  Equation \eqref{eq:FLRWpr} implies that
	\begin{equation} \label{eq:FLRWpr2}
		t(s) \big( p^v(s) - p^u(s) \big) \geq t(s_0) \big( p^v(s_0) - p^u(s_0) \big)
		\qquad
		\text{for all } s \geq s_0,
	\end{equation}
	and so equation \eqref{eq:FLRWpv} implies that
	\begin{multline*}
		t(s) p^v(s)
		-
		t(s_0) p^v(s_0)
		\leq
		\int_{s_0}^s
		\frac{2}{t(s') r(s')^3}
		ds' t(s_0)^2r(s_0)^2 p^u(s_0) p^v(s_0)
		\\
		\leq
		\int_{r(s_0)}^{r(s)}
		\frac{2}{r^3}
		dr \frac{t(s_0)r(s_0)^2 p^u(s_0) p^v(s_0)}{\big( p^v(s_0) - p^u(s_0) \big)}
		\leq
		4 t(s_0) p^v(s_0)
		,
	\end{multline*}
	and the upper bound follows.  The estimate for $p^u(s)$ in \eqref{eq:FLRWcomponents} finally follows from returning to \eqref{eq:FLRWr2pupv}.
	
	Consider now \eqref{eq:FLRWubound}.  The existence of $U_0$ follows trivially from the fact that
	\[
		\dot{u}(s) = p^u(s) \geq 0.
	\]
	The existence of $U_1$ follows from the fact that,
	\[
		u(s) - u(s_0)
		=
		\int_{s_0}^s p^u(s')ds'
		\leq
		\int_{r(s_0)}^{r(s)}
		\frac{C}{t r^2} \frac{p^v(s_0)}{p^v - p^u} dr
		\leq
		\frac{2C}{t(s_0) r(s_0)},
	\]
	by \eqref{eq:FLRWcomponents} and \eqref{eq:FLRWpr2}.
\end{proof}

The proof of Theorem \ref{thm:fixedbackground} can now be given.

\begin{proof}[Proof of Theorem \ref{thm:fixedbackground}]
	Consider Proposition \ref{prop:geodesics} with $\mathcal{B} = \supp(f_1)$.
	The inclusion \eqref{eq:FLRWsuppf} follows immediately from \eqref{eq:FLRWubound}.  Let $v_0$ be sufficiently large so that $R_{\circ}(U_1,v_0)>0$.  The inclusion \eqref{eq:FLRWsuppf} in particular implies that,
	\[
		c r \leq t^{\frac{1}{2}} \leq C r,
		\qquad
		\text{in } \supp(f) \cap \{v \geq v_0\},
	\]
	and so \eqref{eq:FLRWcomponents} imply that
	\[
		\left( \gamma_{AB} p^A p^B \right)^{\frac{1}{2}} \leq \frac{L_0}{t^2},
		\qquad
		0 \leq p^u \leq \frac{Cp^v}{t},
		\qquad
		0 \leq p^v \leq \frac{C}{t},
		\qquad
		c t^2
		\leq 
		R_{\circ}^2
		\leq
		C t^2
		\qquad
		\text{in } \supp(f) \cap \{v \geq v_0\}.
	\]
	Hence, in view of the expression \eqref{eq:emtensordoublenull1},
	\[
		T_{uu}
		=
		\frac{\pi}{2} \frac{\Omega_{\circ}^4}{R_{\circ}^2}
		\int_0^{L_0} \int_0^{Ct^{-1}}  f \, p^v \, dp^v \, L \, dL
		\lesssim
		\frac{\Vert f_1 \Vert_{L^{\infty}}}{t^2}.
	\]
	Similarly, in view of \eqref{eq:emtensordoublenull2} and \eqref{eq:emtensordoublenull3},
	\[
		T_{uv}
		=
		\frac{\pi}{2} \frac{\Omega_{\circ}^4}{R_{\circ}^2}
		\int_0^{L_0} \int_0^{Ct^{-1}}  f \, p^u \, dp^v \, L \, dL
		\lesssim
		\frac{\Vert f_1 \Vert_{L^{\infty}}}{t^3},
		\qquad
		T_{vv}
		=
		\frac{\pi}{2} \frac{\Omega_{\circ}^4}{R_{\circ}^2}
		\int_0^{L_0} \int_0^{Ct^{-1}}  f \, \frac{(p^u)^2}{p^v} \, dp^v \, L \, dL
		\lesssim
		\frac{\Vert f_1 \Vert_{L^{\infty}}}{t^4}.
	\]
\end{proof}

\begin{remark}[Massless Vlasov equation on Minkowski space] \label{rmk:compareMinkowski}
	Consider Minkowski space $(\mathbb{R}^{3+1},m)$, where $m = -dt^2 + (dx^1)^2 + (dx^2)^2 + (dx^3)^2$.  In standard double null coordinates $u= t - r$, $v = t + r$, $r = \sqrt{(x^1)^2+(x^2)^2+(x^3)^2}$, with respect to which the metric takes the form
	\[
		m = - du dv + r^2 \gamma,
	\]
	the components of the energy momentum tensor \eqref{eq:EV4} for solutions of the massless Vlasov equation arising from compactly supported initial data decay with the rates
	\[
		T_{uu} \lesssim \frac{\Vert f_1 \Vert_{L^{\infty}}}{t^2},
		\qquad
		T_{uv} \lesssim \frac{\Vert f_1 \Vert_{L^{\infty}}}{t^4},
		\qquad
		T_{vv} \lesssim \frac{\Vert f_1 \Vert_{L^{\infty}}}{t^6},
	\]
	for $t\geq 1$ \cite{Tay17}. 
	Compare with the rates \eqref{eq:EMtensoronFLRW} which, with respect to the double null frame $e_3 = t^{-\frac{1}{2}} \partial_u$, $e_4 = t^{-\frac{1}{2}} \partial_v$ normalised so that $g_{\circ}(e_3,e_4) = -2$, take the form
	\[
		T_{33} \lesssim \frac{\Vert f_1 \Vert_{L^{\infty}}}{t^3},
		\qquad
		T_{34} \lesssim \frac{\Vert f_1 \Vert_{L^{\infty}}}{t^4},
		\qquad
		T_{44} \lesssim \frac{\Vert f_1 \Vert_{L^{\infty}}}{t^5}.
	\]
\end{remark}

\begin{remark}[Massless Vlasov equation on FLRW with $\mathbb{T}^3$ topology] \label{rmk:T3rates}
	The FLRW metric $g_{\circ}$, along with the function $f_{\circ}$, defined by \eqref{eq:FLRW1}--\eqref{eq:FLRW2}, also define a solution of the Einstein--massless Vlasov system on the manifold $(0,\infty) \times \mathbb{T}^3$.  The components of the energy momentum tensor \eqref{eq:EV4} for solutions of the massless Vlasov equation arising from compactly supported initial data on such a background decay with the rates
	\[
		T_{tt} \lesssim \frac{\Vert f\vert_{\Sigma_1} \Vert}{t^2}, 
		\qquad
		\sum_{i=1}^3 T_{it} \lesssim \frac{\Vert f\vert_{\Sigma_1} \Vert}{t^{\frac{3}{2}}},
		\qquad
		\sum_{i,j=1}^3 T_{ij} \lesssim \frac{\Vert f\vert_{\Sigma_1} \Vert}{t},
	\]
	for all $t$, for a suitable norm $\Vert \cdot \Vert$.  Compare again with the rates \eqref{eq:EMtensoronFLRW}.
\end{remark}

\section{The main theorem}
\label{section:maintheorem}

Theorem \ref{thm:main} can be more precisely stated as follows.  Recall that initial data for the Einstein--massless Vlasov system consists of $(\Sigma, \overline{g},K,f_1)$ (see Section \ref{subsec:Cauchyproblem}).  Let $\Vert \cdot \Vert$ be a norm in which local well posedness and Cauchy stability holds for the Einstein--massless Vlasov system \eqref{eq:EV1}--\eqref{eq:EV4} in spherical symmetry.  The norm on $(\Sigma_1,g_{\circ}\vert_{\Sigma_1})$,
\[
	\Vert (\overline{g},K,f_1) \Vert
	=
	\Vert \overline{g} \Vert_{C^{k+1}}
	+
	\Vert K \Vert_{C^k}
	+
	\Vert f_1 \Vert_{C^{k}},
\]
suffices, for example, for $k$ suitably large.

\begin{theorem}[Future stability of FLRW] \label{thm:main2}
	Let $\mu : [0,\infty) \to [0,\infty)$ be a smooth function, decaying suitably so that $\int_{\mathbb{R}^3} \vert p \vert \mu(\vert p \vert^2) dp$ is finite, such that $\mu \not\equiv 0$.  Let $g_{\circ}$ and $f_{\circ}$ be defined by \eqref{eq:FLRW1}, \eqref{eq:FLRW2}.  Let $(\Sigma_1,\overline{g},K,f_1)$ be a spherically symmetric initial data set for the Einstein--massless Vlasov system on $\Sigma_1 = \{t=1\} \subset \mathcal{M}_{\circ}$ such that $f_1 - f_{\circ}\vert_{\Sigma_1}$ is compactly supported on $P\vert_{\Sigma_1}$ and $\overline{g} - g_{\circ}\vert_{\Sigma_1}$ and $K - \frac{1}{2} \partial_t g_{\circ}\vert_{\Sigma_1}$ are compactly supported on $\Sigma_1$.  There exists $\varepsilon_0>0$ such that, if
	\[
		\Big\Vert \Big( \overline{g} - g_{\circ}\vert_{\Sigma_1}, K - \frac{1}{2} \partial_t g_{\circ}\vert_{\Sigma_1},f_1 - f_{\circ}\vert_{\Sigma_1} \Big) \Big\Vert \leq \varepsilon_0,
	\]
	then the resulting maximal developement of Theorem \ref{thm:Cauchyproblem} is future geodesically complete and there exists $U_0 <0 < U_1$ such that, in an appropriately normalised double null gauge, the associated $(\mathcal{Q},\Omega^2,R,f)$ satisfies the estimates
	\begin{equation} \label{eq:mainthm1}
		\Big\vert \Omega^{-4} R^{-2} T_{uu} - \Omega^{-4}_{\circ} R^{-2}_{\circ} T^{\circ}_{uu} \Big\vert
		\lesssim
		\frac{\varepsilon_0}{t^6},
		\quad
		\Big\vert \Omega^{-2} R^{-4} T_{uv} - \Omega^{-2}_{\circ} R^{-4}_{\circ} T^{\circ}_{uv} \Big\vert
		\lesssim
		\frac{\varepsilon_0}{t^8},
		\quad
		\Big\vert R^{-6} T_{vv} - R^{-6}_{\circ} T^{\circ}_{vv} \Big\vert
		\lesssim
		\frac{\varepsilon_0}{t^{10}},
	\end{equation}
	\begin{equation} \label{eq:mainthm2}
		\vert \Omega^2 - 4t \vert \lesssim \varepsilon_0 t,
		\qquad
		\vert R - t^{\frac{1}{2}} r \vert \lesssim \varepsilon_0,
		\qquad
		\vert \lambda - v \vert \lesssim \varepsilon_0,
		\qquad
		\vert \nu + u \vert \lesssim \varepsilon_0,
	\end{equation}
	in the region $\{U_0 <u < U_1\}$, where $t = \frac{1}{4}(v+u)^2$ and $r=v-u$.  Moreover,
	\[
		\pi(\supp(f - f_{\circ})) \subset \{ U_0 \leq u \leq U_1 \},
	\]
	where $\pi \colon P \to \mathcal{M}$ is the natural projection, and the solution is isometric to the spatially homogeneous FLRW solution \eqref{eq:FLRW1}--\eqref{eq:FLRW2} in the region $\{u \geq U_1\}$.
\end{theorem}

Note that, since the estimates \eqref{eq:mainthm2} imply that $\Omega^4 R^2 \sim t^4$, $\Omega^2 R^4 \sim t^5$ and $R^6 \sim t^6$ in $\supp(f-f_{\circ}) \cap \{v \geq v_0\}$, the rates \eqref{eq:mainthm1} should be viewed as exactly the rates of Theorem \ref{thm:fixedbackground}.

The remainder of the paper is concerned with the proof of Theorem \ref{thm:main2}.

\section{The proof of the main theorem}
\label{section:proof}

This section concerns the proof of Theorem \ref{thm:main2}.  In Section \ref{section:Cauchystab} the proof is reduced so that the main region of consideration arises from an appropriate characteristic initial value problem.  In Section \ref{subsec:renorm} a more convenient renormalisation of the equations of Section \ref{subsec:ssEmVdng} are given.  Section \ref{subsec:bootstrap} concerns Theorem \ref{thm:bootstrap}, which forms the main content of the proof of Theorem \ref{thm:main2}.  In Section \ref{subsec:conclusion} the proof of Theorem \ref{thm:main2} is completed.

\subsection{Cauchy stability and the domain of dependence property}
\label{section:Cauchystab}

Recall the setting of Theorem \ref{thm:main2}.  Let $(\mathcal{M},g)$ denote the unique maximal development of $(\Sigma_1,\overline{g},K,f_1)$, and let $(\mathcal{Q},\Omega^2,R,f)$ denote the associated maximal solution of the reduced system \eqref{eq:Einsteinss1}--\eqref{eq:Vlasovss} (see Section \ref{subsec:ssEmVdng}), in which the residual gauge freedom, discussed in Remark \ref{rmk:residualgauge}, is normalised so that $\lambda \vert_{\{t=1\}} = \lambda_{\circ} \vert_{\{t=1\}}$ and $\nu\vert_{\{t=1\} \cup \Gamma} = \nu_{\circ} \vert_{\{t=1\} \cup \Gamma}$.

Consider some $v_0> U_1$, where $U_1 > 0$ is to be determined later.  Given $\delta_0 >0$, $L_0 \geq 0$, define the set
\[
	\mathcal{U}_{\delta_0,L_0,v_0}
	=
	\{ (x,p)\in P \mid v(x) = v_0, \ 2 - v_0 \leq u(x) \leq v_0-\delta_0, \ 0 \leq p^u \leq 2 p^v, \ L(x,p) \leq L_0 \}
	\subset P,
\]
so that, in particular, $r \geq \delta_0$ and $t \geq 1$ in $\mathcal{U}_{\delta_0,L_0,v_0}$.

By the domain of dependence property, in the region $\{u \leq U_0\}$, for some appropriate $U_0<0$, the solution coincides with the FLRW solution \eqref{eq:FLRW1}--\eqref{eq:FLRW2}.  In particular,
\begin{equation} \label{eq:bcs}
	\Omega^2 = 4t,
	\qquad
	\partial_v \log \Omega^2 = t^{-\frac{1}{2}},
	\qquad
	R = t^{\frac{1}{2}} r,
	\qquad
	\lambda = v,
	\qquad
	\nu = - u,
	\qquad
	\text{on }
	\{u=U_0\}.
\end{equation}

Let $t_0 = t(U_1,v_0) = \frac{1}{4}(v_0+U_1)^2$.  Let $L_0>0$ be such that
\[
	R^4\gamma_{AB} p^A p^B \leq (L_0)^2
	\qquad
	\text{in}
	\qquad
	\supp(f_1 - f_{\circ}\vert_{\Sigma_1}),
\]
and let $\delta_0>0$ be such that $r(U_1,v_0) = v_0 -U_1 > \delta_0$.  By Cauchy stability (for both the Einstein--massless Vlasov system \eqref{eq:EV1}--\eqref{eq:EV4} and the geodesic equations) the solution exists up to $\{t = t_0\}$ and, by Proposition \ref{prop:geodesics}, for any future-maximal null geodesic $\gamma\colon [1,S) \to \mathcal{M}$ such that $(\gamma(1),\dot{\gamma}(1)) \in \supp(f_1)$, there exists $s_0 \in [1,\infty)$ such that
\[
	(\gamma(s_0),\dot{\gamma}(s_0)) \in \mathcal{U}_{\delta_0,L_0,v_0},
\]
if $v_0$ is sufficiently large and $\varepsilon_0>0$ is sufficiently small.  Moreover
\begin{equation} \label{eq:v0}
	\vert \Omega^2 - 4t \vert \lesssim \varepsilon_0,
	\quad
	\vert \partial_u \log \Omega^2 - t^{-\frac{1}{2}} \vert \lesssim \varepsilon_0,
	\quad
	\vert R - t^{\frac{1}{2}} r \vert \lesssim \varepsilon_0,
	\quad
	\vert \lambda - v \vert \lesssim \varepsilon_0,
	\quad
	\vert \nu + u \vert \lesssim \varepsilon_0,
	\quad
	\text{on }
	\{v=v_0\}.
\end{equation}

The remainder of the proof will then be concerned with the region $\{ v\geq v_0 \} \cap \{ u \geq U_0\}$.

\subsection{The renormalised equations}
\label{subsec:renorm}

Rather than considering the equations \eqref{eq:Einsteinss1}--\eqref{eq:Einsteinss4} directly, in the proof of Theorem \ref{thm:main2} it is more convenient to consider the following renormalisation of the equations with their respective FLRW quantities,
\begin{align}
	\partial_u \partial_v (R^2 - R_{\circ}^2)
	&
	=
	-\frac{\Omega^2 - \Omega^2_{\circ}}{2} + (R^2 - R_{\circ}^2) T_{uv} + R_{\circ}^2 (T_{uv} - T^{\circ}_{uv}),
	\label{eq:renom1}
	\\
	\partial_u \partial_v \log \frac{\Omega^2}{\Omega^2_{\circ}}
	&
	=
	\frac{\Omega^2}{2 R^2} - \frac{\Omega^2_{\circ}}{2 R^2_{\circ}}
	+
	\frac{2}{R^2} \nu \lambda - \frac{2}{R^2_{\circ}} \nu_{\circ} \lambda_{\circ}
	-
	2(T_{uv} - T^{\circ}_{uv}),
	\label{eq:renom2}
	\\
	\partial_u(\Omega^{-2} \nu - \Omega^{-2}_{\circ} \nu_{\circ})
	&
	=
	\frac{1}{2} R_{\circ} \Omega^{-2}_{\circ} T^{\circ}_{uu} - \frac{1}{2} R \Omega^{-2} T_{uu},
	\label{eq:renom3}
	\\
	\partial_v(\Omega^{-2} \lambda - \Omega^{-2}_{\circ} \lambda_{\circ})
	&
	=
	\frac{1}{2} R_{\circ} \Omega^{-2}_{\circ} T^{\circ}_{vv} - \frac{1}{2} R \Omega^{-2} T_{vv},
	\label{eq:renom4}
\end{align}
where the solution is compared with the FLRW solution using the double null gauge of Section \ref{subsec:FLRWdoublenull}, as discussed at the end of Section \ref{subsec:FLRWdoublenull}.

For any null geodesic $\gamma\colon [1,S) \to \mathcal{M}$, write the tangent vector to $\gamma$ as
\[
	\dot{\gamma}(s) = p^u(s) \partial_u + p^v(s) \partial_v + p^A (s) \partial_{\theta^A}.
\]
The geodesic equations for $p^u(s)$ and $p^v(s)$ take the form
\begin{align}
	\dot{p}^u(s)
	+
	\partial_u \log \Omega^2 \big( p^u(s)\big)^2
	+
	\frac{2\lambda R}{\Omega^2} \gamma_{AB} p^A p^B
	&
	=
	0,
	\label{eq:geodesicpu}
	\\
	\dot{p}^v(s)
	+
	\partial_v \log \Omega^2 \big( p^v(s)\big)^2
	+
	\frac{2\nu R}{\Omega^2} \gamma_{AB} p^A p^B
	&
	=
	0.
	\label{eq:geodesicpv}
\end{align}
The following renormalisation of \eqref{eq:geodesicpv},
\begin{align} \label{eq:geodesicpvrenom}
	\frac{d}{ds} (tp^v)
	=
	t (t^{-\frac{1}{2}} - \partial_v \log \Omega^2) \big( p^v\big)^2
	+
	\Big( (R - R_{\circ}) - 2t^{\frac{1}{2}}(\nu - \nu_{\circ}) \Big) \frac{t^{\frac{1}{2}} R}{\Omega^2} \gamma_{AB} p^A p^B
	+
	\frac{2t^{\frac{3}{2}} R}{\Omega^2} \gamma_{AB} p^A p^B
	,
\end{align}
will also be used.  Equation \eqref{eq:geodesicpvrenom} is obtained from \eqref{eq:geodesicpv} using the mass shell relation \eqref{eq:massshell} and the fact that $\dot{t} = t^{\frac{1}{2}} (p^v + p^u)$.

\subsection{The bootstrap theorem}
\label{subsec:bootstrap}

The main content of the proof of Theorem \ref{thm:main2} is contained in the following bootstrap theorem.

Given $T >t_0$ and $U_1 >0$, define the region
\begin{equation} \label{eq:QUT}
	\mathcal{Q}_{U_1,T} = \{ (u,v) \in \mathcal{Q} \mid U_0 \leq u \leq U_1, v \geq v_0, \text{and } t_0 \leq t(u,v) < T\}.
\end{equation}

\begin{theorem}[Bootstrap theorem] \label{thm:bootstrap}
	Suppose that $T>t_0$ is such that the maximal development exists up to time $T$ --- in the sense that $(u(t,r),v(t,r)) \in \mathcal{Q}$ for all $t_0 \leq t <T$, $r\geq 0$ --- the difference $f-f_{\circ}$ is supported in the region
	\begin{equation} \label{eq:supportproperty}
		\supp(\pi(f- f_{\circ})) \subset \{U_0 \leq u \leq U_1\},
	\end{equation}
	and, for some fixed $\varepsilon >0$ and $U_1 >0$, the solution moreover satisfies, for all $(u,v) \in \mathcal{Q}_{U_1,T}$,
	\begin{equation} \label{eq:ba1}
		\Big\vert \Omega^{-4} R^{-2} T_{uu} - \Omega^{-4}_{\circ} R^{-2}_{\circ} T^{\circ}_{uu} \Big\vert
		<
		\frac{\varepsilon}{t^6},
		\quad
		\Big\vert \Omega^{-2} R^{-4} T_{uv} - \Omega^{-2}_{\circ} R^{-4}_{\circ} T^{\circ}_{uv} \Big\vert
		<
		\frac{\varepsilon}{t^8},
		\quad
		\Big\vert R^{-6} T_{vv} - R^{-6}_{\circ} T^{\circ}_{vv} \Big\vert
		<
		\frac{\varepsilon}{t^{10}} ,
	\end{equation}
	\begin{equation} \label{eq:ba2}
		\vert \Omega^2 - 4t \vert < \varepsilon t,
		\qquad
		\Big\vert \partial_v \log \Omega^2 - \frac{1}{t^{\frac{1}{2}}} \Big\vert < \frac{\varepsilon}{t},
		\qquad
		\Big\vert \partial_u \log \Omega^2 - \frac{1}{t^{\frac{1}{2}}} \Big\vert < \varepsilon,
	\end{equation}
	\begin{equation} \label{eq:ba3}
		\vert R - t^{\frac{1}{2}} r \vert < \varepsilon,
		\qquad
		\vert R^2 - t r^2 \vert < \varepsilon t,
		\qquad
		\vert \lambda - v \vert < \varepsilon,
		\qquad
		\vert \nu + u \vert < \varepsilon.
	\end{equation}
	Then, if $U_1$ is suitably large and $\varepsilon$ is suitably small, there exists a constant $C>0$ (independent of $\varepsilon$ and $U_1$) such that the inequalities \eqref{eq:ba1}--\eqref{eq:ba3} hold in $\mathcal{Q}_{U_1,T}$ with $\varepsilon$ replaced by $C\varepsilon_0$ and \eqref{eq:supportproperty} holds with $U_1$ replaced by $U_1/2$.  Moreover, the solution is isometric to the spatially homogeneous FLRW solution \eqref{eq:FLRW1}--\eqref{eq:FLRW2} in the region $\{ t_0 \leq t(u,v) < T \} \cap \{u \geq U_1\}$.
\end{theorem}

\begin{proof}
	The proof is divided into several steps.  First, the size of the support of $f - f_{\circ}$ is estimated.  These estimates are then used to obtain estimates on the components of the energy momentum tensor.  The metric quantities are then estimated in the region $\{U_0 \leq u \leq U_1\}$.  Finally, the solution is shown to be isometric to the spatially homogeneous FLRW solution \eqref{eq:FLRW1}--\eqref{eq:FLRW2} in the region $\{u \geq U_1\}$.
	
	Throughout the proof the notation $A\lesssim B$ is used when there exists a constant $K$, which may depend on $U_1-U_0$, such that $A \leq K B$.  Constants $0< c \ll 1 \ll C$ are always independent of $U_1-U_0$.
	
	\noindent \textbf{Estimates for the support of $f - f_{\circ}$:}
	Let $\delta_0$, $L_0$ and $v_0$ be as in Section \ref{section:Cauchystab}, and recall that, for each null geodesic $\gamma\colon [1,S) \to \mathcal{M}$ emanating from $\supp(f_1)$, if $S$ is sufficiently large there exists a time $s_0 \in [1,S)$ such that $(\gamma(s_0),\dot{\gamma}(s_0)) \in \mathcal{U}_{\delta_0,L_0,v_0}$.
	
	For such a null geodesic, write
	\[
		\dot{\gamma}(s) = p^u(s) \partial_u + p^v(s) \partial_v + p^A (s) \partial_{\theta^A}.
	\]
	By the bootstrap assumption \eqref{eq:supportproperty}, the $u$ coordinate of $\gamma$ satisfies
	\[
		U_0 \leq u(s) \leq U_1,
	\]
	for all $s_0 \leq s < S$.  Recall moreover that $v_0>U_1$ and so $R(v_0,U_1)>0$.  It follows in particular that
	\begin{equation} \label{eq:alonggamma}
		t^{\frac{1}{2}} \sim r \sim R^{\frac{1}{2}} \sim v,
		\qquad
		\text{along } \gamma.
	\end{equation}
	There exist constants $0< c \ll 1 \ll C$ such that the components of the tangent vector $\dot{\gamma}$ of any such null geodesic $\gamma\colon [1,S) \to \mathcal{M}$ emanating from $\supp(f_1)$ satisfy
	\begin{equation} \label{eq:components}
		\left( \gamma_{AB} p^A(s) p^B(s) \right)^{\frac{1}{2}} \leq \frac{2 L_0}{t(s) r(s)^2},
		\qquad
		0 \leq p^u(s) \leq \frac{Cp^v(s_0)}{t(s) r(s)^2},
		\qquad
		c \frac{p^v(s_0)}{t(s)} \leq p^v(s) \leq C \frac{p^v(s_0)}{t(s)},
	\end{equation}
	for all $s\geq s_0$.  Indeed, let $\gamma\colon [s_0,S) \to \mathcal{M}$ be a null geodesic such that $(\gamma(s_0),\dot{\gamma}(s_0)) \in \mathcal{U}_{\delta_0,L_0,v_0}$.  The first of \eqref{eq:components} follows from the conservation of angular momentum \eqref{eq:consofam}, which implies that
	\begin{equation} \label{eq:consam}
		R(s)^4 \gamma_{AB} p^A(s) p^B(s)
		=
		R(s_0)^4 \gamma_{AB} p^A(s_0) p^B(s_0)
		\leq
		(L_0)^2,
	\end{equation}
	along with the fact that $t^2r^4 \leq 2R^4$, which follows from the bootstrap assumption \eqref{eq:ba3}.
	
	It follows from the geodesic equation \eqref{eq:geodesicpvrenom} that
	\begin{equation} \label{eq:tpvrenom}
		\frac{d}{ds} \Big( tp^v
		\exp \Big(
		\int_{s_0}^s h(s') p^v(s') ds'
		\Big)
		\Big)
		=
		\exp \Big(
		\int_{s_0}^s h(s') p^v(s') ds'
		\Big)
		\frac{2t^{\frac{3}{2}} R}{\Omega^2} \gamma_{AB} p^A p^B,
	\end{equation}
	where
	\[
		h(s)
		=
		(t^{-\frac{1}{2}} - \partial_v \log \Omega^2)
		+
		\big( (R - R_{\circ}) - 2t^{\frac{1}{2}}(\nu - \nu_{\circ}) \big) \frac{t^{-\frac{1}{2}} R}{\Omega^2} \frac{\gamma_{AB} p^A p^B}{(p^v)^2}.
	\]
	In order to obtain the lower bound of the latter of \eqref{eq:components}, consider the set
	\[
		\mathcal{E} = \left\{ s \in [s_0,S) \mid (tp^v)(s') \geq \frac{1}{2} (tp^v)(s_0) \text{ for all } s_0 \leq s' \leq s \right\} \subset [s_0,S).
	\]
	The set $\mathcal{E}$ is manifestly a closed, connected, non-empty subset of $[s_0,S)$.  Moreover, if $s\in \mathcal{E}$, then
	\begin{equation} \label{eq:inth}
		\int_{s_0}^s \vert h(s') \vert p^v(s') ds'
		=
		\int_{v_0}^{v(s)} \vert h \vert dv
		\leq
		C
		\int_{v_0}^{v(s)}
		\frac{\varepsilon}{v^2}
		dv
		\leq
		C \varepsilon,
	\end{equation}
	by the fact that
	\begin{equation} \label{eq:ROpupv}
		R(s)^4 \gamma_{AB} p^A(s) p^B(s)
		=
		R(s_0)^4 \gamma_{AB} p^A(s_0) p^B(s_0)
		=
		R(s_0)^2 \Omega(s_0)^2 p^u(s_0) p^v(s_0),
	\end{equation}
	along with the bootstrap assumptions \eqref{eq:ba2}, \eqref{eq:ba3} and the equivalence \eqref{eq:alonggamma}.  Since
	\[
		\frac{d}{ds} \Big( tp^v
		\exp \Big(
		\int_{s_0}^s h(s') p^v(s') ds'
		\Big)
		\Big)
		\geq
		0,
	\]
	it follows that
	\[
		(t p^v)(s_0) \leq e^{C \varepsilon} (t p^v)(s),
	\]
	and thus, if $\varepsilon$ is sufficiently small, the set $\mathcal{E} \subset [s_0,S)$ is open, and hence equal to $[s_0,S)$.
	
	For the upper bound of the latter of \eqref{eq:components}, recall equation \eqref{eq:tpvrenom} and note that
	\[
		\int_{s_0}^s \frac{2t^{\frac{3}{2}} R}{\Omega^2} \gamma_{AB} p^A p^B ds'
		=
		\int_{v_0}^{v(s)}
		\frac{2t^{\frac{5}{2}}}{R^3\Omega^2}
		\frac{R^4\gamma_{AB} p^A p^B}{tp^v}
		dv
		\leq
		C\int_{v_0}^{v(s)}
		\frac{p^v(s_0)}{v^3}
		dv
		\leq
		C p^v(s_0),
	\]
	by \eqref{eq:alonggamma}, \eqref{eq:ROpupv} and the lower bound for $(tp^v)(s)$.  The upper bound for $(tp^v)(s)$ then follows from the estimate \eqref{eq:inth}, after integrating equation \eqref{eq:tpvrenom}.
	
	The estimate for $p^u(s)$ in \eqref{eq:components} finally follows from the mass shell relation \eqref{eq:massshell}, which gives
	\[
		\frac{\Omega(s)^2}{t(s)} p^u(s)
		=
		\frac{\big(R^4 \gamma_{AB} p^A p^B\big)(s)}{t(s) p^v(s)},
	\]
	along with the former and latter of \eqref{eq:components}.
	
	It in particular follows from \eqref{eq:components} that
	\begin{equation} \label{eq:insuppfminusflrw}
		\left( \gamma_{AB} p^A p^B \right)^{\frac{1}{2}} \leq \frac{C L_0}{t^2},
		\qquad
		0 \leq p^u \leq \frac{Cp^v}{t},
		\qquad
		0 \leq p^v \leq \frac{C}{t},
		\qquad
		\text{in } \supp(f - f_{\circ}).
	\end{equation}
	
	Consider now the support property \eqref{eq:supportproperty}.  For any null geodesic $\gamma\colon [s_0,S) \to \mathcal{M}$ as above, it follows from \eqref{eq:components} that,
	\[
		u(s) - u(s_0)
		=
		\int_{s_0}^s p^u(s')ds'
		\leq
		\int_{v(s_0)}^{v(s)}
		\frac{C}{t v^2} \frac{p^v(s_0)}{p^v} dv
		\leq
		\frac{C}{t(s_0) v_0}.
	\]
	Thus, since $t(s_0),v_0 \geq 1$,
	\begin{equation} \label{eq:improvedsupport}
		\supp(\pi(f- f_{\circ})) \subset \{U_0 \leq u \leq C\},
	\end{equation}
	and so \eqref{eq:supportproperty} in fact holds with $U_1$ replaced by $U_1/2$, provided that $U_1 \geq 2C$.

	\noindent \textbf{Estimates for the energy momentum tensor components:}
	Recall the expressions \eqref{eq:emtensordoublenull1}--\eqref{eq:emtensordoublenull3} for the components of the energy momentum tensor, along with the expression \eqref{eq:emFLRWdoublenull} for the components of the energy momentum tensor of $f_{\circ}$.  By \eqref{eq:improvedsupport}, for any $(u,v,p^v,L)$,
	\begin{equation} \label{eq:fsupport}
		\vert f(u,v,p^v,L) - f_{\circ}(u,v,p^v,L) \vert \leq \varepsilon_0 \mathbbm{1}_{\{U_0 \leq u \leq U_1/2\}}.
	\end{equation}
	Define
	\[
		L_{\circ}^2 = R_{\circ}^4 \gamma_{AB} p^A p^B,
	\]
	and note that $R^{-4}L^2 = R_{\circ}^{-4} L_{\circ}^2$.  It follows from the support property \eqref{eq:insuppfminusflrw} that
	\begin{align*}
		\Big\vert \Omega^{-4} R^{-2} T_{uu} - \Omega^{-4}_{\circ} R^{-2}_{\circ} T^{\circ}_{uu} \Big\vert
		\leq
		\frac{\pi}{2}
		\int_0^{L_0} \int_0^{Ct^{-1}}
		\vert
		f
		-
		f_{\circ}
		\vert
		\,
		p^{v}
		dp^v R^{-4} L \, dL
		\lesssim
		\frac{\varepsilon_0}{t^6} \mathbbm{1}_{\{U_0 \leq u \leq U_1/2\}}.
	\end{align*}
	Noting that $\Omega^2 R^{-2} p^u = \Omega^2_{\circ} R^{-2}_{\circ} \mathring{p}^u$, it similarly follows that
	\begin{align*}
		\Big\vert \Omega^{-2} R^{-4} T_{uv} - \Omega^{-2}_{\circ} R^{-4}_{\circ} T^{\circ}_{uv} \Big\vert
		\leq
		\frac{\pi}{2}
		\int_0^{L_0} \int_0^{Ct^{-1}}
		\vert
		f
		-
		f_{\circ}
		\vert
		\,
		\frac{\Omega^2 p^u}{R^2}
		dp^v R^{-4} L \, dL
		\lesssim
		\frac{\varepsilon_0}{t^8} \mathbbm{1}_{\{U_0 \leq u \leq U_1/2\}},
	\end{align*}
	and
	\begin{align*}
		\Big\vert R^{-6} T_{vv} - R^{-6}_{\circ} T^{\circ}_{vv} \Big\vert
		&
		\leq
		\frac{\pi}{2}
		\int_0^{L_0} \int_0^{Ct^{-1}}
		\vert
		f
		-
		f_{\circ}
		\vert
		\,
		\frac{\Omega^4(p^u)^2}{R^4 p^{v}}
		dp^v R^{-4} L \, dL
		\lesssim
		\frac{\varepsilon_0}{t^{10}} \mathbbm{1}_{\{U_0 \leq u \leq U_1/2\}}.
	\end{align*}
	Thus
	\begin{align}
		\Big\vert T_{uu} - \frac{1}{t} \Big\vert
		&
		\lesssim
		\frac{\varepsilon_0}{t^2} \mathbbm{1}_{\{U_0 \leq u \leq U_1/2\}}
		+
		t^{-2} \vert R - R_{\circ}\vert
		+
		t^{-2} \vert \Omega^2 - \Omega_{\circ}^2 \vert
		,
		\label{eq:emimprov1}
		\\
		\Big\vert T_{uv} - \frac{1}{2t} \Big\vert
		&
		\lesssim
		\frac{\varepsilon_0}{t^3} \mathbbm{1}_{\{U_0 \leq u \leq U_1/2\}}
		+
		t^{-2} \vert R - R_{\circ}\vert
		+
		t^{-2} \vert \Omega^2 - \Omega_{\circ}^2 \vert
		,
		\label{eq:emimprov2}
		\\
		\Big\vert T_{vv} - \frac{1}{t} \Big\vert
		&
		\lesssim
		\frac{\varepsilon_0}{t^4} \mathbbm{1}_{\{U_0 \leq u \leq U_1/2\}}
		+
		t^{-2} \vert R - R_{\circ}\vert
		.
		\label{eq:emimprov3}
	\end{align}

	\noindent \textbf{Estimates for the metric quantities in the region $u \leq U_1$:}
	For the metric quantities, consider the region $U_0 \leq u \leq U_1$ and recall that
	\[
		t^{\frac{1}{2}} \sim r \sim R^{\frac{1}{2}} \sim v, \qquad \text{for } U_0 \leq u \leq U_1, \quad v \geq v_0.
	\]
	Let $u$ and $v$ be such that $U_0 \leq u \leq U_1$, $v \geq v_0$ and $t_0 \leq t(u,v) < T$.
	
	Note that
	\[
		\frac{1}{2} \partial_v (R^2 - R_{\circ}^2) = R_{\circ} (\lambda - \lambda_{\circ} \big) + (R - R_{\circ}) \lambda.
	\]
	Equation \eqref{eq:renom1} and the estimate \eqref{eq:emimprov2} then imply that
	\[
		\big\vert \partial_u \big( R_{\circ} (\lambda - \lambda_{\circ} \big) + (R - R_{\circ}) \lambda \big\vert
		\lesssim
		\vert \Omega^2 - \Omega_{\circ}^2 \vert
		+
		\vert R - R_{\circ} \vert
		+
		\varepsilon_0 v^{-2}.
	\]
	Thus, integrating and using the boundary condition \eqref{eq:bcs},
	\begin{equation} \label{eq:regioni1}
		\vert \lambda - \lambda_{\circ} \vert(u,v)
		\lesssim
		v^{-1} \vert R - R_{\circ} \vert(u,v)
		+
		v^{-2} \int_{U_0}^u
		\vert \Omega^2 - \Omega_{\circ}^2 \vert(u',v)
		+
		\vert R - R_{\circ} \vert (u',v)
		du'
		+
		\varepsilon_0 v^{-4}.
	\end{equation}
	Similarly, since
	\[
		\nu - \nu_{\circ}
		=
		\Omega^2
		\big( \Omega^{-2} \nu - \Omega^{-2}_{\circ} \nu_{\circ} \big)
		-
		\nu_{\circ} (1 - \Omega^{2} \Omega^{-2}_{\circ})
		,
	\]
	it follows from equation \eqref{eq:renom3} and the estimate \eqref{eq:emimprov2} and the boundary condition \eqref{eq:bcs} that
	\begin{equation} \label{eq:regioni2}
		\vert \nu - \nu_{\circ}\vert(u,v)
		\lesssim
		v^{-2} \vert \Omega^2 - \Omega_{\circ}^2 \vert (u,v)
		+
		v^{-2} \int_{U_0}^u
		\vert \Omega^2 - \Omega_{\circ}^2 \vert(u',v)
		+
		\vert R - R_{\circ} \vert (u',v)
		du'
		+
		\varepsilon_0 v^{-2}.
	\end{equation}
	Moreover,
	\[
		\partial_u (R - R_{\circ}) = \nu - \nu_{\circ},
	\]
	and so \eqref{eq:regioni2} and the Gr\"{o}nwall inequality imply that
	\begin{equation} \label{eq:regioni3}
		\vert R - R_{\circ}\vert(u,v)
		\lesssim
		v^{-2} \int_{U_0}^u
		\vert \Omega^2 - \Omega_{\circ}^2 \vert(u',v)
		du'
		+
		\varepsilon_0 v^{-2}.
	\end{equation}
	Returning to \eqref{eq:regioni1} and \eqref{eq:regioni2}, one then has
	\begin{equation} \label{eq:regioni4}
		\vert \lambda - \lambda_{\circ} \vert(u,v)
		\lesssim
		v^{-2} \int_{U_0}^u
		\vert \Omega^2 - \Omega_{\circ}^2 \vert(u',v)
		du'
		+
		\varepsilon_0 v^{-3},
	\end{equation}
	and
	\begin{equation} \label{eq:regioni5}
		\vert \nu - \nu_{\circ}\vert(u,v)
		\lesssim
		v^{-2} \vert \Omega^2 - \Omega_{\circ}^2 \vert (u,v)
		+
		v^{-2} \int_{U_0}^u
		\vert \Omega^2 - \Omega_{\circ}^2 \vert(u',v)
		du'
		+
		\varepsilon_0 v^{-2}
		.
	\end{equation}
	Now, by equation \eqref{eq:renom2} and the estimate \eqref{eq:emimprov2}
	\begin{equation} \label{eq:regioni6}
		\Big\vert \partial_u \partial_v \log \frac{\Omega^2}{\Omega^2_{\circ}} \Big\vert
		\lesssim
		v^{-4} \vert \Omega^2 - \Omega_{\circ}^2 \vert
		+
		v^{-4} \vert R - R_{\circ} \vert
		+
		v^{-3} \vert \nu - \nu_{\circ} \vert
		+
		v^{-4} \vert \lambda - \lambda_{\circ} \vert
		+
		\varepsilon_0 v^{-4},
	\end{equation}
	Integrating in $u$ and $v$, using the boundary condition \eqref{eq:bcs} and \eqref{eq:v0}, and inserting \eqref{eq:regioni3}--\eqref{eq:regioni5}, it follows that
	\begin{equation} \label{eq:regioni7}
		\Big\vert \log \frac{\Omega^2}{\Omega^2_{\circ}} \Big\vert
		\lesssim
		\varepsilon_0
		+
		\int_{v_0}^v \int_{U_0}^u
		(v')^{-4} \vert \Omega^2 - \Omega_{\circ}^2 \vert(v',u')
		du'dv'
	\end{equation}
	Now since, for $\varepsilon$ suitably small,
	\[
		\vert \Omega^2 - \Omega_{\circ}^2 \vert
		\lesssim
		\Omega^2 \Big\vert \log \frac{\Omega^2}{\Omega^2_{\circ}} \Big\vert
		\lesssim
		v^2 \Big\vert \log \frac{\Omega^2}{\Omega^2_{\circ}} \Big\vert
		\qquad
		\text{for } U_0 \leq u \leq U_1,
	\]
	it follows from the Gr\"{o}nwall inequality, after integrating \eqref{eq:regioni7} once more in $u$, that
	\[
		\vert \Omega^2 - \Omega_{\circ}^2 \vert
		\lesssim
		\varepsilon_0 v^2.
	\]
	Returning to \eqref{eq:regioni3}--\eqref{eq:regioni5}, it then follows that
	\[
		\vert R - R_{\circ}\vert
		\lesssim
		\varepsilon_0,
		\qquad
		\vert \lambda - \lambda_{\circ} \vert
		\lesssim
		\varepsilon_0,
		\qquad
		\vert \nu - \nu_{\circ}\vert
		\lesssim
		\varepsilon_0.
	\]
	Similarly, returning to \eqref{eq:regioni6},
	\[
		\Big\vert \partial_v \log \Omega^2 - \frac{1}{t^{\frac{1}{2}}} \Big\vert \lesssim \varepsilon_0 v^{-2},
		\qquad
		\Big\vert \partial_u \log \Omega^2 - \frac{1}{t^{\frac{1}{2}}} \Big\vert \lesssim \varepsilon_0.
	\]
	
	\noindent \textbf{The region $u \geq U_1$:}
	In order to see that the solution is isometric to the spatially homogeneous FLRW solution \eqref{eq:FLRW1}--\eqref{eq:FLRW2} in the region $u \geq U_1$, recall first \eqref{eq:fsupport}, from which it follows that
	\[
		\big( R^{-6}T_{vv} \big) ( U_1,v) = \big(R_{\circ}^{-6} T^{\circ}_{vv}\big) (U_1,v) = t(U_1,v)^{-4} r(U_1,v)^{-6} \qquad \text{for all } v \geq V_1,
	\]
	where $V_1$ is such that $R(U_1,V_1) = 0$.  Recall the residual gauge freedom of Remark \ref{rmk:residualgauge} and consider the following change of gauge.  Define increasing functions $\widetilde{v}(v)$ and $\widetilde{u}(u)$ as solutions of
	\[
		\widetilde{v}(V_1) = U_1,
		\qquad
		\widetilde{v}'(v) = \lambda(U_1,v) \lambda_{\circ}^{-1}(U_1,\widetilde{v}(v)),
		\qquad
		\text{for } v \geq V_1,
	\]
	and
	\[
		\widetilde{u}(u) = u
		\quad
		\text{for } u \leq U_1,
		\qquad
		\widetilde{u}' (u) = \frac{1}{\widetilde{v}'(V_1)} \Omega^2(u,V_1) \Omega_{\circ}^{-2}(\widetilde{u}(u),V_1)
		\quad
		\text{for } u > U_1,
	\]
	respectively.  It follows that $\{u=U_1\} = \{ \widetilde{u} = U_1\}$, and moreover the relations \eqref{eq:gauge1} and \eqref{eq:gauge4} imply that the solution in the resulting gauge satisfies
	\[
		R (U_1,U_1) = 0,
		\qquad
		\Omega^2(U_1,U_1) = \Omega_{\circ}^2(U_1,U_1)
		\qquad
		\partial_vR (U_1,v) = \partial_vR_{\circ}(U_1,v),
		\qquad
		\text{for all } v \geq U_1.
	\]
	The Raychaudhuri equation \eqref{eq:Einsteinss4}, restricted to $u=U_1$, takes the form
	\[
		\partial_v(\Omega^{-2} \partial_v R)(U_1,v)
		=
		- \frac{1}{2} R^7 \Omega^{-2} \big( R_{\circ}^{-6} T^{\circ}_{vv} \big)(U_1,v).
	\]
	It then follows that
	\[
		R (U_1,v) = R_{\circ}(U_1,v),
		\qquad
		\Omega^2(U_1,v) = \Omega_{\circ}^2(U_1,v)
		\qquad
		\text{for all } v \geq U_1.
	\]
	By uniqueness for the characteristic initial value problem, it follows that $R = R_{\circ}$ and $\Omega^2 = \Omega^2_{\circ}$ for all $u \geq U_1$.
\end{proof}

\subsection{The conclusion of the proof of the main theorem}
\label{subsec:conclusion}

Recall the set $\mathcal{Q}_{U_1,T} \subset{Q}$ defined by \eqref{eq:QUT}.  Consider the set $\mathcal{A} \subset [t_0, \infty)$ of times $T \in [t_0, \infty)$ such that the solution exists up to time $T$ --- in the sense that $(u(t,r),v(t,r)) \in \mathcal{Q}$ for all $t_0 \leq t <T$, $r\geq 0$ --- and moreover satisfies the estimates \eqref{eq:ba1}--\eqref{eq:ba3} for all $(u,v) \in \mathcal{Q}_{U_1,T}$ for some $\varepsilon>0$ sufficiently small so that the conclusion of Theorem \ref{thm:bootstrap} holds.  The set $\mathcal{A} \subset [t_0,\infty)$ is manifestly connected and open, and is moreover non-empty by local existence and Cauchy stability, provided $\varepsilon_0$ is sufficiently small.  By Theorem \ref{thm:noncentext}, the bounds \eqref{eq:ba2}--\eqref{eq:ba3} imply that $(u(T,r),v(T,r)) \in \mathcal{Q}$ for all $r\geq 0$ such that $u(T,r) \leq U_1$.  By Theorem \ref{thm:bootstrap}, the solution is isometric to the spatially homogeneous FLRW solution \eqref{eq:FLRW1}--\eqref{eq:FLRW2} in the region $\{ t_0 \leq t(u,v) < T \} \cap \{u \geq U_1\}$ and thus $(u(T,r),v(T,r)) \in \mathcal{Q}$ for all $r\geq 0$.  Theorem \ref{thm:bootstrap} moreover guarantees that the bounds \eqref{eq:ba1}--\eqref{eq:ba3} cannot be saturated, and thus the set $\mathcal{A}$ is moreover closed, and hence equal to $[t_0,\infty)$.

The future geodesic completeness follows from the bounds on the Christoffel symbols arising from \eqref{eq:ba2}--\eqref{eq:ba3} (as, for example in Section 16 of \cite{LiRo}).

\bibliography{masslessVlasovFLRW}{}
\bibliographystyle{plain}

\end{document}